\documentclass[a4paper,10pt,openright]{article} 
\usepackage[english]{babel} 
\usepackage[italian]{}
\usepackage[latin1]{inputenc} 
\usepackage{amsmath,amssymb, stmaryrd,amsthm,amsfonts,amscd,color,enumerate,eucal,latexsym,mathrsfs,mathtools,cancel,tikz}
\usepackage{epsfig, hieroglf, protosem}  
\usepackage{simplewick, bm}
\usepackage{hyperref}
\usepackage[all,cmtip]{xy}
\usetikzlibrary{matrix,calc}
\usepackage{verbatim}
\numberwithin{equation}{section}
\newtheoremstyle{TheoremStyle}
{3pt}
{3pt}
{\slshape}
{}
{\bf}
{:}
{.5em}
{}

\definecolor{PaColor}{RGB}{255,50, 150}
\definecolor{ClaColor}{RGB}{0,0,255}

\newcommand{\bigslant}[2]{{\raisebox{.2em}{$#1$}\left/\raisebox{-.2em}{$#2$}\right.}}

\theoremstyle{TheoremStyle}
\newtheorem{theorem}{Theorem}
\newtheorem{corollary}[theorem]{Corollary}
\newtheorem{proposition}[theorem]{Proposition}

\newtheorem{definition}[theorem]{Definition}
\newtheorem{remark}[theorem]{Remark}
\newtheorem{example}[theorem]{Example} 

\title{Besov Wavefront Set}

\author{Claudio Dappiaggi \thanks{CD: Dipartimento di Fisica,
	Universit\`a degli Studi di Pavia \& INFN, Sezione di Pavia, 
	Via Bassi 6,
	I-27100 Pavia,
	Italia; Istituto Nazionale di Alta Matematica, Sezione di Pavia, via Ferrata 5, 27100 Pavia, Italia
	\mbox{claudio.dappiaggi@unipv.it}
}
\and
	Paolo Rinaldi \thanks{PR: Institute for Applied Mathematics, Universit\"at Bonn, 
	Endenicher Allee 60,
	D-53115 Bonn,
	Germany;
	\mbox{rinaldi@iam.uni-bonn.de}
}
\and
	Federico Sclavi \thanks{FS: Dipartimento di Fisica,
	Universit\`a degli Studi di Pavia \& INFN, Sezione di Pavia, 
	Via Bassi 6,
	I-27100 Pavia,
	Italia; Istituto Nazionale di Alta Matematica, Sezione di Pavia, via Ferrata 5, 27100 Pavia, Italia
\mbox{federico.sclavi01@universitadipavia.it}}
}

\begin{document}
\maketitle
\begin{abstract}
We develop a notion of wavefront set aimed at characterizing in Fourier space the directions along which a distribution behaves or not as an element of a specific Besov space. Subsequently we prove an alternative, albeit equivalent characterization of such wavefront set using the language of pseudodifferential operators. Both formulations are used to prove the main underlying structural properties. Among these we highlight the individuation of a sufficient criterion to multiply distributions with a prescribed Besov wavefront set which encompasses and generalizes the classical Young's theorem. At last, as an application of this new framework we prove a theorem of propagation of singularities for a large class of hyperbolic operators.
\end{abstract}


\section{Introduction}

Microlocal analysis and the associated H\"ormander's wavefront set \cite{Hor94} have been an unmitigated success in analysis which has found in addition manifold applications ranging from engineering to mathematical physics. One of the most recent interplay with modern theoretical physics is related to the r\^{o}le played by microlocal techniques in the construction of a full-fledged theory of quantum fields on generic Lorentzian and Riemannian backgrounds as well as in the development of a mathematical formulation of renormalization with the language of distributions, see {\it e.g.} \cite{Brunetti-Fredenhagen-00, Rejzner:2016hdj, DDR20,CDDR20}.

In the early developments of the interplay between microlocal analysis and renormalization, it has become clear that the original framework developed by H\"ormander aimed at disentangling the directions of rapid decrease in Fourier space of a given distribution from the singular ones suffered from a substantial limitation. As a matter of fact, in many concrete scenarios one is interested in having a more refined estimate of the singular behavior of a distribution, for instance comparing it with that of an element lying in a suitable Sobolev space. This has lead to considering more specific forms of wavefront set, among which a notable r\^{o}le in application has been played by the so-called Sobolev wavefront set, see \cite{Hor97}. Still having in mind the realm of quantum field theory, one of the first remarkable uses has been discussed in \cite{Junker:2001gx}, while nowadays it has become an essential ingredient in many modern results among which noteworthy are those concerning the analysis of the wave equations on manifolds with boundaries or with corners, see {\it e.g.} \cite{Vasy08,Vasy12}. 

An apparently completely detached branch of analysis in which distributions and their specific singular behavior plays a distinguished r\^{o}le is that of stochastic partial differential equations. Without entering in too many technical details, far from the scope of this work, remarkable leaps forward have been obtained in the past few years both within the framework of the theory of regularity structures \cite{Hairer14, Hairer15} and in that of paracontrolled distributions \cite{Gubinelli}. In both approaches, despite the necessity of dealing with specific problems, such as renormalization, calling for the analysis of products or of extensions of a priori ill-defined distributions, microlocal techniques never enter the game. 

The reasons are manifold but the main one lies in the fact that, in the realm of stochastic partial differential equations, often one considers H\"older distributions, {\it i.e.} elements of $C^\alpha(\mathbb{R}^d)\subset\mathcal{S}^\prime(\mathbb{R}^d)$, $\alpha\in\mathbb{R}$. The latter can be read as a specific instance of the so-called Besov spaces $B^\alpha_{p,q}(\mathbb{R}^d)$, $\alpha\in\mathbb{R}$, $p,q\in [1,\infty]$, \cite{Tri06}. When working in this framework, one relies often in Bony paradifferential calculus \cite{Bony} as it is devised to better catch the specific features of elements lying in a Besov space. To this end microlocal techniques and the wavefront set in particular appear at first glance to be far from the optimal tool to be used, since it appears to be unable to grasp the peculiar singular behaviour of a distribution in comparison to an element of $B^\alpha_{p,q}(\mathbb{R}^d)$. 

Nonetheless it has recently emerged that, in the analysis of a large class of nonlinear stochastic partial differential equations, microlocal analysis can be used efficiently to devise a recursive scheme to construct both solutions and correlation functions, while taking into account intrinsically the underlying renormalization freedoms, \cite{DDRZ20, BDR21}. One of the weak point of this novel approach lies in the lack of any control of the convergence of the underlying recursive scheme. This can be ascribed mainly to the fact that employing microlocal techniques appears to wash out all information concerning the behaviour of the underlying distributions as elements of a Besov space. Observe that each $B^\alpha_{p,q}(\mathbb{R}^d)$ is endowed with the structure of a Banach space which is pivotal in setting up a fixed point argument to prove the existence of solutions for the considered class of nonlinear stochastic partial differential equations. 

Hence, it appears natural to seek a way to combine the best of both worlds, trying to use the language of microlocal analysis on the one hand, while keeping track of the underlying Besov space structure on the other hand. In this paper we plan to make the first step in this direction, developing a modified notion of wavefront set, specifically devised to keep track of the behaviour of a distribution in comparison to that of an element of a Besov space. For definiteness and in order to avoid unnecessary technical difficulties, focusing instead on the main ideas and constructions, we shall focus on the Besov spaces $B^\alpha_{\infty,\infty}(\mathbb{R}^d)\equiv C^\alpha(\mathbb{R}^d)$, which are, moreover, the most relevant ones in concrete applications. We highlight that an investigation in this direction, complementing our own, has appeared in \cite{GM15}.

Specifically our proposal hinges on the following starting point, a definition of {\em Besov wavefront set}
which focuses on the behaviour of a distribution in Fourier space. 

\begin{definition}\label{def:besov wavefront set_intro}
	Let $u \in \mathcal{D}^\prime(\mathbb{R}^d)$ and $\alpha \in \mathbb{R}$. 
	We say that $(x,\xi)\in \mathbb{R}^d \times (\mathbb{R}^d\setminus \{0\})$ does not lie in the $B_{\infty,\infty}^{\alpha}$-wavefront set, denoting $(x,\xi) \not \in \mathrm{WF}^\alpha(u)$, if there exist $\phi\in\mathcal{D}(\mathbb{R}^d)$ with $\phi(x)\neq 0$ as well as an open conic neighborhood $\Gamma$ of $\xi$ in $\mathbb{R}^d\setminus\{0\}$ such that
	\begin{align}\label{eq:wavefront set besov_intro}
		\bigg \lvert \int_{\Gamma} \widehat{\phi u}(\eta) \check{\underline{\kappa}}(\eta) e^{i y \cdot \eta} d\eta \bigg \rvert \lesssim& 1\,, 	
	\end{align}
	\begin{align}\label{eq:wavefront set besov II_intro}
		\bigg \lvert \int_{\Gamma} \widehat{\phi u}(\eta) \check{\kappa}(\lambda\eta) e^{i y \cdot \eta} d\eta \bigg \rvert \lesssim& \lambda^\alpha\,,
	\end{align}
	for any $\underline{\kappa}\in\mathcal{D}(B(0,1))$ with $\check{\underline{\kappa}}(0)\neq 0$, $\lambda \in (0,1)$, $y \in \operatorname{supp}(\phi)$ and $\kappa\in \mathscr{B}_{\lfloor \alpha \rfloor}$, see Definition \ref{Def: Kernel for Besov}. 
\end{definition}

While conceptually the above definition enjoys all desired structural properties, from an operational viewpoint, it is rather difficult to use it concretely both in examples and in the proof of various results. For this reason we give an alternative, albeit equivalent, characterization of $\mathrm{WF}^\alpha(u)$, $u\in\mathcal{D}^\prime(\mathbb{R}^d)$, in terms of the intersection of the characteristic set of a suitable class of order zero, properly supported pseudodifferential operators, see Proposition \ref{prop: characterization via pseudodifferential operators}. Using this tool we are able to prove a large set of structural properties of the Besov wavefront set. The three main results that we obtain are the following:
\begin{itemize}
	\item We prove that, given an embedding $f\in C^\infty(\Omega,\Omega^\prime)$ between two open subsets $\Omega\subseteq\mathbb{R}^d$ and $\Omega^\prime\subseteq\mathbb{R}^m$, one can establish a criterion, see Theorem \ref{th: pullback theorem 2}, for the existence of the pull-back $f^*u$, $u\in\mathcal{D}^\prime(\Omega^\prime)$ which generalizes the one devised by H\"ormander in the smooth setting, \cite[Thm. 8.2.4]{Hor94}. A noteworthy byproduct of this analysis is that, whenever $f$ is a diffeomorphism, then, for any $\alpha\in\mathbb{R}$, $f^*WF^\alpha(u)=WF^\alpha(f^*u)$, see Theorem \ref{th:product theorem 2}. This result is noteworthy since it entails that the notion of Besov wavefront set can be applied also to distributions supported on an arbitrary smooth manifold \cite{RS21}.
	\item We establish a sufficient criterion for the existence of the product of two distributions with prescribed Besov wavefront set and we provide an estimate for the wavefront set of the product, see Theorem \ref{th:product theorem 2}. This result contains and actually extends the renown Young's theorem on the product of two H\"older distributions, which is often used in the applications to stochastic partial differential equations.
	\item We apply the whole construction of the Besov wavefront set to prove a propagation of singularities theorem for a large class of hyperbolic partial differential equations, see Theorem \ref{th:propagation of singularities}. This result is strongly tied to a preliminary analysis on the wavefront set $\mathrm{WF}^\alpha(\mathcal{K}(u))$ where $\mathcal{K}$ is a linear map from $C^\infty_0(\Omega^\prime)\to\mathcal{D}^\prime(\Omega)$ where $\Omega\subseteq\mathbb{R}^d$ while $\Omega^\prime\subseteq\mathbb{R}^m$.
\end{itemize}

The paper is organized as a follows: In Section \ref{Sec: Preliminaries}, we present the definition of Besov spaces outlining some of its main properties and alternative, equivalent characterizations. Subsequently we review succinctly the basic notions of pseudodifferential operators and of the associated operator wavefront set. In Section \ref{Sec: Besov Wavefront Set} we present the main object of our investigation, giving the definition of Besov wavefront set in terms of the behaviour of a distribution in Fourier space, outlining subsequently some of the basic structural properties and discussing a few notable examples. In Section \ref{Sec: PsiDOWF} we prove that the Besov wavefront set can be equivalently characterized in terms of the characteristic set of a suitable class of properly supported pseudodifferential operators. Section \ref{Sec: Operations on BWF} contains the main results concerning the structural properties of the Besov wavefront set. In particular we discuss its interplay with pullbacks, we devise a sufficient criterion for the product of two distributions with prescribed Besov wavefront set and we prove a theorem of propagation of singularities for a class of hyperbolic partial differential operators.

\paragraph{Notations}
In this short paragraph we fix a few recurring notations used in this manuscript. With $\mathcal{E}(\mathbb{R}^d)$ ({\em resp.} $\mathcal{D}(\mathbb{R}^d))$, we denote the space of smooth ({\em resp.} smooth and compactly supported) functions on $\mathbb{R}^d$, $d\geq 1$, while $\mathcal{S}(\mathbb{R}^d)$ stands for the space of rapidly decreasing smooth functions. 
Their topological dual spaces are denoted respectively $\mathcal{E}^\prime(\mathbb{R}^d)$, $\mathcal{D}^\prime(\mathbb{R}^d)$ and $\mathcal{S}^\prime(\mathbb{R}^d)$. 
In addition, given $u\in\mathcal{S}(\mathbb{R}^d)$, we adopt the following convention to define its Fourier transform 
\begin{align*}
\mathcal{F}(u)(k)=\hat{u}(k)\vcentcolon=\int_{\mathbb{R}^d}e^{-ik\cdot x}u(x)\,dx\,.
\end{align*}
At the same time, we indicate with the symbol $\check{\cdot}$ the inverse Fourier transform $\mathcal{F}^{-1}$, namely, for any $f\in\mathcal{S}(\mathbb{R}^d)$, $f=\check{\hat{f}}=\hat{\check{f}}$.
Similarly, for any $v\in\mathcal{S}^\prime(\mathbb{R}^d)$, we indicate with $\widehat{v}\in\mathcal{S}^\prime(\mathbb{R}^d)$ its Fourier transform, defining it per duality as $
\hat{v}(u)\doteq v(\hat{u})$ for all $u\in\mathcal{S}(\mathbb{R}^d)$. 
In general, given a function $f\in\mathcal{E}(\mathbb{R}^d)$, $x\in\mathbb{R}^d$ and $\lambda\in(0,1]$, we shall denote $f^\lambda_x(y):=\lambda^{-d}f(\lambda^{-1}(y-x))$.
At last with $\langle x\rangle\vcentcolon=(1+|x|^2)^{\frac{1}{2}}$ we denote the Japanese bracket, while the symbol $\lesssim$ refers to an inequality holding true up to a multiplicative finite constant. Observe that, depending on the case in hand, such constant might depend on other data, such as for example the choice of an underlying compact set. For the ease of notation we shall omit making such dependencies explicit, since they shall become clear from the context.

\section{Preliminaries}\label{Sec: Preliminaries}
The aim of this section is to introduce the key function spaces and some of their notable properties. The content of this specific subsection is mainly inspired by \cite{BCD11, Tri06}. The starting point lies in the notion of a \emph{Littlewood-Paley} partition of unity.
\begin{definition}\label{Def: L-P partition of unity}
Let $N\in\mathbb{N}$ and let $\psi\in\mathcal{D}(\mathbb{R}^d)$ be a positive function supported in $\{2^{-N}\leq|\xi|\leq2^N\}$. 
We call \textit{Littlewood-Paley} partition of unity a sequence $\{\psi_j\}_{j\in\mathbb{N}_0}$, $\mathbb{N}_0\doteq\mathbb{N}\cup\{0\}$ such that
\begin{itemize}
\item $\psi_0\in\mathcal{D}(\mathbb{R}^d)$ and $\textrm{supp}(\psi_0)\subseteq\{|\xi|\leq2^N\}$;	
\item $\psi_j(x)\vcentcolon=\psi(2^{-j}x)$ for $j\geq1$;
\item $\sum\limits_{j\in\mathbb{N}_0}\psi_j(\xi)=1$ for all $\xi\in\mathbb{R}^d$;
\item for any multi-index $\alpha$, $\exists C_\alpha>0$ such that
\begin{align*}
|D^\alpha\psi_j(\xi)|\leq C_\alpha\langle\xi\rangle^{-|\alpha|}\,,\qquad j\geq1\,;
\end{align*}
\item $\psi_j(-\xi)=\psi_j(\xi)$ for all $j\geq 0$.
\end{itemize}
\end{definition}
\noindent
In the following we shall always assume for definiteness $N=1$.

\begin{definition}\label{def:definition Besov Fourier approach}
	Let $\alpha\in\mathbb{R}$. We call {\bf Besov space} $B^\alpha_{p,q}(\mathbb{R}^d)$, $p,q\in[1,\infty)$, the Banach space whose elements $u$ are such that
	\begin{equation}\label{eq:besov norm_pq}
		\|u\|^q_{B_{pq}^\alpha(\mathbb{R}^d)}\vcentcolon=\sum\limits_{j\geq 0}2^{j\alpha q} \|\psi_j(D)u\|^q_{L^p(\mathbb{R}^d)}<\infty\,,
	\end{equation}
At the same time if $q=\infty$, while $p\in [1,\infty]$, we set
	\begin{equation}\label{eq:besov norm}
		\|u\|_{B_{p,\infty}^\alpha(\mathbb{R}^d)}\vcentcolon=\sup_{j\geq 0}2^{j\alpha} \|\psi_j(D)u\|_{L^p(\mathbb{R}^d)}<\infty\,,
	\end{equation}
where we used the Fourier multiplier notation $\psi_j(D)u(x)\vcentcolon=\mathcal{F}^{-1}\{\psi_j(\xi)\hat{u}(\xi)\}(x)$. At the same time, we say that $u \in B_{\infty,\infty}^{\alpha,\mathrm{loc}}(\mathbb{R}^d)$ if $\varphi u \in B_{\infty,\infty}^{\alpha}(\mathbb{R}^d)$ for any $\varphi \in \mathcal{D}(\mathbb{R}^d)$.
\end{definition}

\begin{remark}
By definition of Fourier multiplier, it descends that
\begin{align*}
 \psi_j(D) u (x) = \mathcal{F}^{-1}\{\psi_j(\xi) \hat{u}(\xi)\}(x) = (\check{\psi}_j \ast u)(x)= u(2^{jd}\check{\psi}(2^j(\cdot-x)))\,,
\end{align*}
	where we exploited $\mathcal{F}^{-1}\{uv\}= \check{u} \ast \check{v}$ and $\check{\psi}_j(x)=2^{jd}\check{\psi}(2^j x)$. 
As a consequence, if $u \in B_{\infty,\infty}^\alpha(\mathbb{R}^d)$,
	\begin{equation}\label{eq:besov scaling}
		\lvert u(\check{\psi}^{2^{-j}}_x) \rvert \lesssim 2^{-j\alpha},\quad \forall j\geq 0, \quad \forall x \in \mathbb{R}^d\,.
	\end{equation}
\end{remark}

In our analysis it will be often convenient not to consider directly Definition \ref{def:definition Besov Fourier approach}, rather to work with an equivalent characterization, dubbed the \emph{local means formulation} -- see \cite[Sec. 1.4 \& Thm.1.10]{Tri06}. This is based on the following tool.

\begin{definition}\label{Def: Kernel for Besov}
Let $B(0,1)=\{y \in \mathbb{R}^d: \lvert y \rvert < 1\}$. For $s\in\mathbb{N}_0$, we call $\mathscr{B}_s$ the subset of $\mathcal{D}(B(0,1))$ whose elements $\kappa$ are such that there exists $\epsilon>0$
\begin{equation}\label{eq:conditions kappa}
\check{\kappa}(\xi)\neq0\;\text{ if }\;\frac{\varepsilon}{2}<\lvert \xi \rvert \leq 2\varepsilon\,,\quad\textrm{and}\quad(\partial^\beta\check{\kappa})(0)=0 \,\;\textrm{if}\;\lvert\beta\rvert\leq s\,.
\end{equation}
\end{definition}

\noindent Observe that the second condition in Equation \eqref{eq:conditions kappa} is empty if $s<0$.

\begin{definition}\label{def:definition Besov local means}
	Let $\alpha\in\mathbb{R}$, $\kappa\in\mathscr{B}_{\lfloor\alpha\rfloor}$, with $\lfloor\alpha\rfloor$ the biggest integer $N$ such that $N\leq\alpha$. Let $\underline{\kappa}\in\mathcal{D}(B(0,1))$ be such that $\check{\underline{\kappa}}(0)\neq0$. 
	We call $B_{p,\infty}^\alpha(\mathbb{R}^d)$, $p\in [1,\infty]$, the space of distributions $u\in\mathcal{S}^\prime(\mathbb{R}^d)$ such that
	\begin{equation}\label{eq:besov local means}
		\|u\|^{\kappa, \underline{\kappa}}_{B_{p,\infty}^\alpha(\mathbb{R}^d)} \vcentcolon= \|u(\underline{\kappa}_{x})\|_{L^p(\mathbb{R}^d)} + \sup_{\lambda \in (0,1)} \frac{\|u(\kappa^\lambda_{x})\|_{L^p(\mathbb{R}^d)}}{\lambda^\alpha} < \infty\,,
	\end{equation}
where the $L^\infty$-norm is taken with respect to the variable $x$.
\end{definition}


\begin{remark}
	We observe that different choices for $\kappa$ and $\underline{\kappa}$ yield in Equation \eqref{eq:besov local means} equivalent norms. Therefore, henceforth we shall omit to indicate the superscripts $\kappa$ and $\underline{\kappa}$.
\end{remark}

\noindent If $\alpha < 0$, there exists a further equivalent characterization for Besov spaces -- see \cite[Prop. A.5]{BL21}, \cite[Cor. 1.12]{Tri06}. We focus on the case $p=\infty$.

\begin{proposition}\label{prop:negative besov}
	Let $\alpha < 0$ and $\kappa \in \mathcal{D}(B(0,1))$ be such that $\check{\kappa}(0)\neq 0$. 
	Then $u \in B_{\infty,\infty}^\alpha(\mathbb{R}^d)$ if and only if
	\begin{equation}\label{eq:local means negative Besov spaces}
		\sup_{\lambda \in (0,1)} \frac{\|u(\kappa^\lambda_{x})\|_{L^\infty(\mathbb{R}^d)}}{\lambda^\alpha} < \infty\,,
	\end{equation}
where the $L^\infty$-norm is taken with respect to the variable $x$.
\end{proposition}


\noindent We conclude this subsection proving a last, useful characterization of the element lying in $B_{\infty,\infty}^\alpha$.
\begin{proposition}\label{Prop: fourier besov}
Let $ u \in \mathcal{S}^\prime(\mathbb{R}^d)$ and let $\alpha \in \mathbb{R}$. 
Then $u \in B_{\infty,\infty}^{\alpha}(\mathbb{R}^d)$ if and only if, given $\kappa \in \mathscr{B}_{\lfloor \alpha\rfloor}$ and $\underline{\kappa}\in \mathcal{D}(B(0,1))$ such that $\check{\underline{\kappa}}(0)\neq 0$, it holds that
\begin{equation}\label{eq:fourier besov distribution}
	\lvert \langle \hat{u}(\xi), e^{ix\cdot \xi} \check{\underline{\kappa}}(\xi) \rangle \rvert \lesssim 1\,, \qquad \lvert \langle \hat{u}(\xi), e^{ix\cdot \xi} \check{\kappa}(\lambda\xi) \rangle \rvert \lesssim \lambda^\alpha\,,
\end{equation}
for any $\lambda \in (0,1)$ and $x \in \mathbb{R}^d$. 
\end{proposition}
\begin{proof}
The statement is a direct consequence of Definition \ref{def:definition Besov local means} combined with the following identities
\begin{align}\label{eq:fourier transform scaling}
		u(\varphi_x) = \langle \hat{u}(\xi), e^{ix\cdot \xi} \check{\varphi}(\xi) \rangle\,, \qquad u(\varphi^\lambda_x) = \langle \hat{u}(\xi), e^{ix\cdot \xi} \check{\varphi}(\lambda\xi) \rangle\,,
	\end{align}
where $\varphi\in\mathcal{S}(\mathbb{R}^d)$, $u \in \mathcal{S}^\prime(\mathbb{R}^d)$, $x\in\mathbb{R}^d$ and $\lambda\in(0,1]$.
In turn these are a by-product of the identities $
u(\varphi) =  \hat{u}(\check{\varphi})\,$, and $\check{(\varphi^\lambda_x)}(\xi) = e^{ix\cdot \xi} \check{\varphi}(\lambda\xi)$.
\end{proof}

\begin{remark}\label{Rmk: fourier besov}
	Observe that, if $\alpha < 0$, then it is sufficient to verify the second of the two conditions in Equation \eqref{eq:fourier besov distribution}.
\end{remark}

\subsection{Pseudodifferential Operators}
In this section we shall focus on the second functional tool which plays a distinguished r\^{o}le in our analysis. Hence we recall succinctly the definition and some notable properties of pseudodifferential operators.
 For later convenience, this section is mainly inspired by \cite{Hintz},  though further details can be found in \cite{Grigis, Hor94}. We start by recalling the definition both of a symbol and of its quantization.
\begin{definition}\label{Def: Symbols}
	Let $m \in \mathbb{R}$ and $n,N \in \mathbb{N}$. A function $a \in C^\infty(\mathbb{R}^d\times \mathbb{R}^N)$ is called a {\bf symbol} of order $m$ if, for all $\alpha \in \mathbb{N}^n_0$, $\beta \in \mathbb{N}^N_0$, it satisfies
	\begin{equation}\label{eq: bound symbols}
		\lvert \partial^\alpha_x \partial_\xi^\beta a(x,\xi) \rvert \leq C_{\alpha\beta} \langle \xi \rangle^{m-\lvert \beta \rvert}
	\end{equation}
	for some constant $C_{\alpha\beta} > 0$ and for any $x$ in a compact set of $\mathbb{R}^d$. 
	We denote the space of symbols of order $m$ with $S^m(\mathbb{R}^d;\mathbb{R}^N)$. 
	In addition, we define the space of residual symbols by
	\begin{equation}\label{eq: residual symbols}
		S^{-\infty}(\mathbb{R}^d;\mathbb{R}^N) := \bigcap_{m \in \mathbb{R}} S^m(\mathbb{R}^d;\mathbb{R}^N).
	\end{equation}
At last we call $S^m_{\textrm{hom}}(\mathbb{R}^d;\mathbb{R}^N)\subset S^m(\mathbb{R}^d;\mathbb{R}^N)$ the collection of homogeneous symbols of order $m$, namely, when $|\xi|>1$,  $a(x,\lambda\xi)=\lambda^m a(x,\xi)$ for all $\lambda>0$ and, for all $\alpha \in \mathbb{N}^n_0$, $\beta \in \mathbb{N}^N_0$
$$\lvert \partial^\alpha_x \partial_\xi^\beta a(x,\xi) \rvert \leq C_{\alpha\beta} |\xi|^{m-\lvert \beta \rvert}.$$
\end{definition}

\begin{definition}\label{Def: Symbols and PsiDO} Let $m \in \mathbb{R}$, $n\in\mathbb{N}$ and let $a \in S^m(\mathbb{R}^d \times \mathbb{R}^d;\mathbb{R}^n)$. 
	We define its quantization $\mathrm{Op}(a):\mathcal{S}(\mathbb{R}^d)\to\mathcal{S}(\mathbb{R}^d)$ as
	\begin{equation}\label{eq: quantization}
		(\mathrm{Op}(a)u)(x) := (2\pi)^{-n} \int_{\mathbb{R}^d} \int_{\mathbb{R}^d} e^{i(x-y) \cdot \xi} a(x,y,\xi) u(y) dy d\xi, \quad u \in \mathcal{S}(\mathbb{R}^d).
	\end{equation}
$\mathrm{Op}(a)$ is called a \textbf{pseudodifferential operator $\bm{\Psi}$DO of order} $\bm{m}$ and the whole set of these operators is denoted by $\Psi^m(\mathbb{R}^d)$.
Moreover, we set
\[
\Psi^{-\infty}(\mathbb{R}^d) := \bigcap_{m \in \mathbb{R}} \Psi^m(\mathbb{R}^d).
\]
\end{definition}

\noindent

Since it plays a r\^{o}le in our analysis, we remark that Equation \eqref{eq: quantization} can be replaced either by the \emph{right quantization} $\mathrm{Op}_R(a)$ or by the {\em left quantization} $\mathrm{Op}_L(a)$
\begin{subequations}
\begin{equation}\label{Eq: right quantization}
	(\mathrm{Op}_R(a^\prime)u)(x) := (2\pi)^{-n} \int_{\mathbb{R}^d} \int_{\mathbb{R}^d} e^{i(x-y) \cdot \xi} a^\prime(y,\xi) u(y)\, dy d\xi.\quad\forall a^\prime\in S^m(\mathbb{R}^d;\mathbb{R}^d)
\end{equation}
\begin{equation}\label{Eq: left quantization}
		(\mathrm{Op}_L(\tilde{a})u)(x) := (2\pi)^{-n} \int_{\mathbb{R}^d} \int_{\mathbb{R}^d} e^{i(x-y) \cdot \xi} \tilde{a}(x,\xi) u(y) dy d\xi\quad \quad\forall \tilde{a}\in S^m(\mathbb{R}^d;\mathbb{R}^d)
\end{equation}
\end{subequations}

\noindent It is important to stress that, at the level of pseudodifferential operators, the choices of quantization procedure is to a certain extent immaterial, since, for any $a\in S^m(\mathbb{R}^d \times \mathbb{R}^d;\mathbb{R}^d)$, there always exist $a_L,a_R\in S^m(\mathbb{R}^d;\mathbb{R}^d)$ such that -- see \cite[Thm. 4.8]{Hintz}
$$\mathrm{Op}(a)=\mathrm{Op}_L(a_L)=\mathrm{Op}_R(a_R).$$

\begin{remark}\label{Rem: extension of PsiDOs}
	By means of a standard duality argument one can extend continuously the action of a pseudodifferential operator of order $m$, $m\in\mathbb{R}$, to tempered distributions. In order not to burdening the reader with an unnecessarily baroque notation, we still indicate any such extension as $\mathrm{Op(a)}:\mathcal{S}^\prime(\mathbb{R}^d)\to\mathcal{S}^\prime(\mathbb{R}^d)$ for  all $a\in S^m(\mathbb{R}^d \times \mathbb{R}^d;\mathbb{R}^d)$.
\end{remark}

\noindent As a last step we give a characterization of a notable subclass of pseudodifferential operators, based on their support properties.

\begin{definition}\label{Def: Properly Supported PsiDO}
	Let $A \in \Psi^m(\mathbb{R}^d)$ and let $K_A \in \mathcal{S}^\prime(\mathbb{R}^d \times \mathbb{R}^d)$ be the associated Schwartz kernel. We say that $A$ is {\em properly supported} if the canonical projections $\pi_1 \colon \mathrm{supp}(K)\subseteq\mathbb{R}^d \times \mathbb{R}^d \to \mathbb{R}^d$ and $\pi_2 \colon \mathrm{supp}(K)\subseteq\mathbb{R}^d \times \mathbb{R}^d \to \mathbb{R}^d$  are proper maps.
\end{definition}

\noindent Associated to a pseudodifferential operator, one can introduce the notion of \emph{operator wavefront set}, which is a key ingredient in our construction outlined in Section \ref{Sec: Besov Wavefront Set}.

\begin{definition}\label{Def: essential support}
	Let $a \in S^m(\mathbb{R}^d;\mathbb{R}^N)$.  We say that a point $(x_0,\xi_0) \in \mathbb{R}^d \times (\mathbb{R}^N \setminus \{0\})$ does not lie in the \textbf{essential support} of $a$
	\[
		\textrm{ess supp}(a) \subset \mathbb{R}^d \times (\mathbb{R}^N\setminus \{0\})\,,
	\]
	if there exists $\varepsilon > 0$ such that for all $\ell \in \mathbb{N}_0^n$, $\beta \in \mathbb{N}^N_0$, $k \in \mathbb{R}$, it holds 
	\begin{equation}
		\lvert \partial^\ell_x \partial^\beta_\xi a(x,\xi) \rvert \leq C \langle \xi \rangle^{-k}, \quad \forall (x,\xi),\;\textrm{such that}\; \lvert \xi \rvert \geq 1,\;\textrm{and}\; \lvert x-x_0\rvert + \bigg \lvert \frac{\xi}{\lvert \xi\rvert} - \frac{\xi_0}{\lvert \xi_0 \rvert} \bigg \rvert < \varepsilon.
	\end{equation}
\end{definition}

Observe that $\textrm{ess supp}(a)$ is a closed subset of $\mathbb{R}^d \times (\mathbb{R}^N \setminus \{0\})$ whereas, for each $x\in\mathbb{R}^d$, $\pi_\xi[\textrm{ess supp}(a)]\subseteq\mathbb{R}^N\setminus\{0\}$ is a conical subset. At last we can state the main definition of this whole section:

\begin{definition}\label{Def: Operator WF}
	Let $A = \mathrm{Op}_L(a)\in\Psi^m(\mathbb{R}^d)$. 
	The \textbf{operator wavefront set} of $A$ is  
	\begin{equation}\label{eq:operator wavefront set}
		WF^\prime(A) := \textrm{ess supp}(a) \subset \mathbb{R}^d \times (\mathbb{R}^N \setminus \{0\}).
	\end{equation}
\end{definition}

\noindent In the following proposition we summarize a few notable properties of the operator wave set. Since the proof is a direct application of Definition \ref{Def: essential support} and \ref{Def: Operator WF}, we omit it.

\begin{proposition}\label{Prop: WF Properties}
	Let $A, B \in \Psi^m(\mathbb{R}^d)$. The following properties hold:
	\begin{itemize}
		\item[(1)] If $A$ has compactly supported Schwartz kernel, then $WF^\prime(A) = \emptyset$ if and only if $A \in\Psi^{-\infty}(\mathbb{R}^d)$.
		\item[(2)] $WF^\prime(A + B) \subset WF^\prime(A) \cup WF^\prime(B)$.
		\item[(3)] $WF^\prime(AB) \subset WF^\prime(A) \cap WF^\prime(B)$.
		\item[(4)] $WF^\prime(A^\ast) = WF^\prime(A)$, where $A^\ast$ is the adjoint of $A$ defined so that for all $u,v\in\mathcal{S}(\mathbb{R}^d)$
		$$\int\limits_{\mathbb{R}^d}dx\,(A^*u)(x)\overline{v}(x)=\int\limits_{\mathbb{R}^d}dx\,u(x)\overline{(Av)(x)}.$$
	\end{itemize}
\end{proposition}

\noindent A further concept, related to $\Psi$DOs and of great relevance in the following sections is that of microlocal parametrix. 
Here we recall its construction. Without entering into many details, for which we refer in particular to \cite[Chap. 3]{Grigis}, we underline that, given any $A\in\Psi^m(\mathbb{R}^d)$, $m\in\mathbb{R}$, one can always associate to it a principal symbol $[\sigma_m(A)]\in\bigslant{S^m(\mathbb{R}^d;\mathbb{R}^d)}{S^{m-1}(\mathbb{R}^d;\mathbb{R}^d)}$. 
In the following, when we do not write explicitly the square brackets, we are considering a representative within the equivalence class identifying the principal symbol. 

\begin{definition}\label{def: elliptic set}
	Given $A \in \Psi^m(\mathbb{R}^d)$, a point $(x_0,\xi_0) \in \mathbb{R}^d \times (\mathbb{R}^d \setminus \{0\})$ does not lie in the \textbf{elliptic set} of $A$, $\mathrm{Ell}(A)$, if there exists $\varepsilon > 0$ and a constant $C>0$ such that
	\begin{equation}
		\lvert \sigma_m(A)(x,\xi) \rvert \geq C | \xi |^{m}, \quad \forall (x,\xi)\;\textrm{such that}\; \lvert \xi \rvert \geq 1,\;\textrm{and}\; \lvert x-x_0\rvert + \bigg \lvert \frac{\xi}{\lvert \xi\rvert} - \frac{\xi_0}{\lvert \xi_0 \rvert} \bigg \rvert < \varepsilon,
	\end{equation} 
where $[\sigma_m(A)]$ is the principal symbol of $A$. We call {\bf characteristic set} of $A$, $\mathrm{Char}(A)$, the complement of $\mathrm{Ell}(A)$.
\end{definition}


\begin{remark}\label{remark: equivalent def elliptic set}
	Definition \ref{def: elliptic set} can be reformulated as follows: a point $(x_0,\xi_0)\in\mathrm{Ell}(A)$ if there exist $b \in S^{-m}(\mathbb{R}^d;\mathbb{R}^d)$ and a conic neighbourhood of $(x_0,\xi_0)$ such that therein $P_m(A)b-1 \in S^{-1}(\mathbb{R}^d;\mathbb{R}^d)$.
\end{remark}

\begin{proposition}\label{prop: parametrix}
	Let $A \in \Psi^m(\mathbb{R}^d)$ and let $\mathscr{C} \subset \mathrm{Ell}(A)$ be a closed subset. Then there exists $B \in \Psi^{-m}(\mathbb{R}^d)$  such that
	\begin{equation}\label{eq: condition parametrix}
		\mathscr{C} \cap WF^\prime(AB-I) = \emptyset, \quad \mathscr{C} \cap WF^\prime(BA-I) = \emptyset.
	\end{equation}
	$B$ is called \textbf{microlocal parametrix} for $A$ on $\mathscr{C}$.
\end{proposition}

\noindent The proof of this proposition can be found in \cite[Prop. 6.15]{Hintz}. For later convenience we conclude the section stating a result on the properties of pseudodifferential operators acting on Besov spaces, see \cite[Sect 6.6]{Abe12}.

\begin{theorem}\label{th: boundedness on besov spaces}
	Let $m \in \mathbb{R}$, $\alpha \in \mathbb{R}$ and let $a\in S^m(\mathbb{R}^d;\mathbb{R}^d)$. Let $A:\mathcal{S}^\prime(\mathbb{R}^d)\to\mathcal{S}^\prime(\mathbb{R}^d)$ be the associated element of $\Psi^m(\mathbb{R}^d)$ as per Definition \ref{Def: Symbols and PsiDO}, Equation \eqref{Eq: left quantization} and Remark \ref{Rem: extension of PsiDOs}. Then the restriction of $A$ to a Besov space as per Definition \ref{def:definition Besov Fourier approach} setting $p=q=\infty$ is a bounded linear operator $A:B^\alpha_{\infty,\infty}(\mathbb{R}^d) \to B^{\alpha - m}_{\infty,\infty}(\mathbb{R}^d)$.
\end{theorem}

\subsubsection{Localization of a $\Psi$DO}\label{Sec: Localization of a PsiDO}

In the next sections, we will be interested in the behaviour of $\Psi$DOs under the action of a local diffeomorphism.  To this end we adapt to our framework and to our notations the analysis in \cite[Chap. 18.1]{Hor94}. 

Hence, let $\Omega\subset\mathbb{R}^d$ be an open subset, we say that a function $a\in C^\infty(\Omega\times\mathbb{R}^d)$ identifies a {\em local symbol} on $\Omega\times\mathbb{R}^d$, {\it i.e.} $a\in S^m(\Omega;\mathbb{R}^d)$ if $\phi a\in S^m(\mathbb{R}^d;\mathbb{R}^d)$ for all $\phi\in C^\infty_0(\Omega)$. Using Equation \eqref{Eq: left quantization} one identifies an operator 
\begin{equation}\label{Eq: local PsiDO}
\mathrm{Op}_L(a):\mathcal{S}^\prime(\mathbb{R}^d)\to\mathcal{D}^\prime(\Omega).
\end{equation}
Observing that $C^\infty_0(\Omega)\hookrightarrow\mathcal{E}^\prime(\Omega)\hookrightarrow\mathcal{S}^\prime(\mathbb{R}^d)$, one can restrict the domain in Equation \eqref{Eq: local PsiDO} to an operator $\mathrm{Op}_L(a):\mathcal{E}^\prime(\Omega)\to\mathcal{D}^\prime(\Omega)$ or $\mathrm{Op}_L(a):C^\infty_0(\Omega)\to C^\infty(\Omega)$, where with a slight abuse of notation we keep on using the same symbol $\mathrm{Op}_L(a)$. In full analogy with Definition \ref{Def: Symbols and PsiDO}, we indicate the ensuing collection of pseudodifferential operators by $\Psi^m(\Omega)$. The following theorem is the direct adaptation to our setting and notations of \cite[Thm. 18.1.17]{Hor94}.

\begin{theorem}\label{th: psido coordinate change}
	Let $\Omega,\Omega^\prime \subset \mathbb{R}^d$ be open subsets, $f\in\mathrm{Diff}(\Omega;\Omega^\prime)$ and let $A \in \Psi^m(\Omega^\prime)$. Then 
	\begin{equation}
		A_f \colon C_0^\infty(\Omega) \to C^\infty(\Omega), \quad u\mapsto A_f u :=  A((f^{-1})^\ast u) \circ f
	\end{equation}
	is a pseudodifferential operator of order $m$. Moreover, 
	\begin{equation}\label{eq: coordinate change principal symbol}
		\sigma_m(A_f)(x,\xi) = \sigma_m(A)(f(x), ({}^t df(x))^{-1} \xi),
	\end{equation}
	where $\sigma_m(A_f)$ and $\sigma_m(A)$ are the principal symbols of $A_f$ and $A$ respectively while $df$ stands for the differential map associated with $f$.
\end{theorem}

\section{Besov Wavefront Set}\label{Sec: Besov Wavefront Set}
The aim of this section is to introduce our main object of investigation. We shall therefore give a definition of \emph{Besov wavefront set}, discussing subsequently its main structural properties. We proceed in two different, albeit ultimately equivalent ways. The first is based on the prototypical notion of wavefront set based on Fourier transforms -- \cite[Ch. 8]{Hor03}, while the second, outlined in Section \ref{Sec: PsiDOWF}, relies on pseudodifferential operators as introduced in Section \ref{Sec: Besov Wavefront Set}.  Observe that, in the following, we rely heavily on Proposition \ref{Prop: fourier besov} as well as on Definition \ref{Def: Kernel for Besov}.

\begin{definition}\label{def:besov wavefront set}
	Let $u \in \mathcal{D}^\prime(\mathbb{R}^d)$ and $\alpha \in \mathbb{R}$. 
	We say that $(x_0,\xi_0)\in \mathbb{R}^d \times (\mathbb{R}^d\setminus \{0\})$ does not lie in the $B_{\infty,\infty}^{\alpha}$-wavefront set, denoting $(x_0,\xi_0) \not \in \mathrm{WF}^\alpha(u)$, if there exist $\phi\in\mathcal{D}(\mathbb{R}^d)$ with $\phi(x)\neq 0$ as well as an open conic neighborhood $\Gamma$ of $\xi$ in $\mathbb{R}^d\setminus\{0\}$ such that for any compact set $K \subset \mathbb{R}^d$
	\begin{align}\label{eq:wavefront set besov}
		\bigg \lvert \int_{\Gamma} \widehat{\phi u}(\xi) \check{\underline{\kappa}}(\xi) e^{i x \cdot \xi} d\xi \bigg \rvert \lesssim& 1\,, 	
	\end{align}
	\begin{align}\label{eq:wavefront set besov II}
		\bigg \lvert \int_{\Gamma} \widehat{\phi u}(\xi) \check{\kappa}(\lambda\xi) e^{i x \cdot \xi} d\xi \bigg \rvert \lesssim& \lambda^\alpha\,,
	\end{align}
for any $\kappa\in \mathscr{B}_{\lfloor \alpha \rfloor}$, $\underline{\kappa}\in\mathcal{D}(B(0,1))$ with $\check{\underline{\kappa}}(0)\neq 0$, $\lambda \in (0,1]$ and $x \in K$.
\end{definition}

\begin{remark}
Observe that, on account of Proposition \ref{Prop: fourier besov} and of Remark \ref{Rmk: fourier besov}, whenever $\alpha<0$ in Definition \ref{def:besov wavefront set} it suffices to check that Equation \eqref{eq:wavefront set besov II} holds true.
\end{remark}

We are now in a position to prove some basic properties of the Besov wavefront set which are a direct consequence of its definition.

\begin{proposition}\label{Prop: empty Besov}
	Let $u\in\mathcal{D}^\prime(\mathbb{R}^d)$. 
	Then
	\begin{align*}
	u\in B_{\infty,\infty}^{\alpha,\mathrm{loc}}(\mathbb{R}^d) \iff \mathrm{WF}^\alpha(u) = \emptyset\,.
	\end{align*}
\end{proposition}
\begin{proof}
	The implication
\begin{align*}
	u \in B_{\infty,\infty}^{\alpha,\mathrm{loc}}(\mathbb{R}^d) \Longrightarrow \mathrm{WF}^\alpha(u) = \emptyset\,,
\end{align*}
	follows immediately combining Definition \ref{def:definition Besov Fourier approach} and Proposition \ref{Prop: fourier besov} with Definition \ref{def:besov wavefront set}.
	Conversely, if $\mathrm{WF}^\alpha(u) = \emptyset$, then once more Definition \ref{def:besov wavefront set} entails that, for any $\phi\in\mathcal{D}(\mathbb{R}^d)$, it holds
	\begin{align*}
		\bigg \lvert \int_{\mathbb{R}^d} \widehat{\phi u}(\eta) e^{iy\cdot \eta} \check{\kappa}(\lambda \eta) d\eta \bigg\rvert \lesssim \lambda^\alpha\,, \qquad \bigg \lvert \int_{\mathbb{R}^d} \widehat{\phi u}(\eta) e^{iy\cdot \eta} \check{\underline{\kappa}}(\eta) d\eta \bigg\rvert \lesssim 1\,.
	\end{align*}
From Proposition \ref{Prop: fourier besov} it descends that $\phi u\in B_{\infty,\infty}^\alpha(\mathbb{R}^d)$ for any $\phi\in\mathcal{D}(\mathbb{R}^d) \nonumber$. 
This proves the sought statement.
\end{proof}

\begin{proposition}
	Let $u,v \in \mathcal{D}^\prime(\mathbb{R}^d)$. Then
	\[
	\mathrm{WF}^\alpha(u+v) \subset \mathrm{WF}^\alpha(u) \cup \mathrm{WF}^\alpha(v).
	\]
\end{proposition}
\begin{proof}
	Assume $(x_0,\xi_0) \in \mathrm{WF}^\alpha(u+v)$. Then, for any test function $\phi \in \mathcal{D}(\mathbb{R}^d)$, open conic neighborhood $\Gamma$ of $\xi_0$, there exists a compact set $K\subset \mathbb{R}^d$ such that, for any $N \in\mathbb{N}$, it holds true
	\[
	\bigg \lvert \int_{\Gamma} \widehat{\phi (u+v)}(\xi) \check{\kappa}(\overline{\lambda}\xi) e^{i \overline{x} \cdot \xi} d\xi \bigg \rvert > N {\overline{\lambda}}^\alpha,
	\]
for some $\overline{x} \in K$  and $\overline{\lambda} \in (0,1]$. Applying the triangle inequality, it descends
	\[
	N {\overline{\lambda}}^\alpha < \bigg \lvert \int_{\Gamma} \widehat{\phi u}(\xi) \check{\kappa}(\overline{\lambda}\xi) e^{i \overline{x} \cdot \xi} d\xi \bigg \rvert + \bigg \lvert \int_{\Gamma} \widehat{\phi v}(\xi) \check{\kappa}(\overline{\lambda}\xi) e^{i \overline{x} \cdot \xi} d\xi \bigg \rvert,
	\]
	which entails that $(x_0,\xi_0) \in \mathrm{WF}^\alpha(u)\cup \mathrm{WF}^\alpha(v)$.
\end{proof}

\begin{corollary}
Let $u \in \mathcal{D}^\prime(\mathbb{R}^d)$. If $\alpha_1 \leq \alpha_2$, then
\begin{equation}\label{eq:inclusion besov spaces}
	\mathrm{WF}^{\alpha_1}(u) \subseteq \mathrm{WF}^{\alpha_2}(u).
\end{equation}
\end{corollary}
\begin{proof}
	The inclusion in Equation \eqref{eq:inclusion besov spaces} follows immediately from Definition \ref{def:besov wavefront set}, particularly Equation \eqref{eq:wavefront set besov II}.
\end{proof}

\begin{remark}\label{Rem: Besov WF for smooth functions}
	Observe that, on account of the inclusion $C^\infty(\mathbb{R}^d)\subset B^{\alpha,\mathrm{loc}}_{\infty,\infty}(\mathbb{R}^d)$ for all $\alpha\in\mathbb{R}$, Proposition \ref{Prop: empty Besov} entail that, for every $f\in C^\infty(\mathbb{R}^d)$
	$$\mathrm{WF}^\alpha(f)=\emptyset.\quad\forall\alpha\in\mathbb{R}$$
	In particular, this result entails that, given any $u\in\mathcal{D}^\prime(\mathbb{R}^d)$, if $x\notin\textrm{singsupp}(u)$, then $(x,\xi)\notin\mathrm{WF}^\alpha(u)$ for all $\alpha\in\mathbb{R}$. Here $\mathrm{singsupp}(u)$ refers to the singular support of $u$, see \cite[Def. 2.2.3]{Hor03} for the definition.
\end{remark}

\noindent In the following, we give some explicit examples of Besov wavefront sets. Observe that the results of Remark \ref{Rem: Besov WF for smooth functions} are always implicitly taken into account.

\begin{example}\label{Ex: BWF of delta}
	Let $u = \delta\in\mathcal{D}^\prime(\mathbb{R}^d)$ be the Dirac delta centered at the origin. Recalling that for any $\phi\in\mathcal{D}(\mathbb{R}^d)$ $\phi\delta=\phi(0)\delta$, Equation \eqref{eq:wavefront set besov} translates to 
	$$\left|\int_\Gamma\underline{\check{\kappa}}(\eta)e^{iy\cdot\eta}d\eta\right|\leq \int_\Gamma\left|\underline{\check{\kappa}}(\eta)\right|d\eta\lesssim 1,$$
since $\underline{\check{\kappa}}\in\mathcal{S}(\mathbb{R}^d)$. Here we have neglected $\phi(0)$ since it plays no r\^{o}le. Focusing instead on Equation \eqref{eq:wavefront set besov II}, for any choice of $\phi\in\mathcal{D}(\mathbb{R}^d)$ with $\phi(0)\neq 0$, it descends, neglecting once more $\phi(0)$, that 
$$\left|\int_\Gamma\check{\kappa}(\eta)e^{iy\cdot\eta}d\eta\right|\leq\int_\Gamma\left|\check{\kappa}(\lambda\eta)\right|d\eta\lesssim\lambda^{-d},$$
where the last inequality descends from the change of variable $\eta\mapsto\eta^\prime:=\lambda\eta$. While this estimate entails that $\mathrm{WF}^\alpha(\delta)=\emptyset$ if $\alpha\leq -d$, in order to obtain a sharp estimate
observe that we can set $y=0$ in Equation \eqref{eq:wavefront set besov II} since it lies in $\mathrm{supp}(\phi)$ for any admissible $\phi$, being $\phi(0)\neq 0$. Hence it descends
$$\left|\int_\Gamma\check{\kappa}(\lambda\eta)d\eta\right|=\lambda^{-d}\left|\int_\Gamma\check{\kappa}(\eta^\prime)d\eta^\prime\right|=C_{\check{\kappa}}\lambda^{-d},$$
where $\eta^\prime:=\lambda\eta$ and where we used implicitly both that $\Gamma$ is a cone and that $\check{\kappa}\in\mathcal{S}(\mathbb{R}^d)$. At this stage, comparing with Definition \ref{def:besov wavefront set}, we can conclude that
	\begin{equation}
	\mathrm{WF}^\alpha(\delta)=
	\left\{\begin{array}{ll}
		\emptyset & \alpha \leq -d, \\
		(0,\xi):\xi \in \mathbb{R}^d \setminus \{0\} & \alpha > -d. \nonumber
	\end{array}\right.
\end{equation}
\end{example}

\begin{example}
	Let $u = \partial_j \delta\in\mathcal{D}^\prime(\mathbb{R}^d)$ be a derivative of the Dirac delta centered at the origin, {\it i.e.} $\partial_j=\frac{\partial}{\partial x_j}$, $x_j$ being an Euclidean coordinate on $\mathbb{R}^d$. Following Definition \ref{def:besov wavefront set} and using the identity $\phi\partial_j\delta=\phi(0)\partial_j\delta-(\partial_j\phi)(0)\delta$ for any $\phi\in\mathcal{D}(\mathbb{R}^d)$, Equation \eqref{eq:wavefront set besov} translates to 
	$$\left|(\partial_j\phi)(0)\int_\Gamma\eta_j\underline{\check{\kappa}}(\eta)e^{iy\cdot\eta}d\eta-\phi(0)\int_\Gamma\underline{\check{\kappa}}(\eta)e^{iy\cdot\eta}d\eta\right|\leq\int\limits_\Gamma\left|(\partial_j\phi)(0)\eta_j-\phi(0))\underline{\check{\kappa}}(\eta)\right|d\eta\lesssim 1,$$
where, similarly to Example \ref{Ex: BWF of delta}, we exploited that $\underline{\check{\kappa}}\in\mathcal{S}(\mathbb{R}^d)$. Focusing on Equation \eqref{eq:wavefront set besov II}, we can repeat the same procedure as in Example \ref{Ex: BWF of delta}. For the sake of conciseness we focus directly only on $y=0$ since it lies in $\textrm{supp}(\phi)$, for any $\phi\in\mathcal{D}(\mathbb{R}^d)$ with $\phi(0)\neq 0$. In addition we can consider only the contribution due to $\phi(0)\partial_j\delta$ which yields, omitting $\phi(0)$ for simplicity of the notation,
$$\left|\int\limits_\Gamma \eta_j\check{\kappa}(\lambda\eta)d\eta\right|=\lambda^{-d-1}\left|\int\limits_\Gamma \eta^\prime_j\check{\kappa}(\eta^\prime)d\eta^\prime\right|=\widetilde{C}_{\check{\kappa}}\lambda^{-d-1},$$
where $\eta^\prime:=\lambda\eta$ and where we used implicitly both that $\Gamma$ is a cone and that $\check{k}\in\mathcal{S}(\mathbb{R}^d)$. Adding to this equality the outcome of Example \ref{Ex: BWF of delta}, it descends
	\begin{equation}
	\mathrm{WF}^\alpha(\partial_j\delta)=
	\left\{\begin{array}{ll}
		\emptyset & \alpha \leq -d-1, \\
		(0,\xi):\xi \in \mathbb{R}^d \setminus \{0\} & \alpha > -d-1. \nonumber
	\end{array}\right.
\end{equation}	
\end{example}

\begin{example}\label{Ex: Besov WF of compactly supported distributions}
	Let $u \in \mathcal{E}^\prime(\mathbb{R}^d)$. Observe that there exists $C > 0$ such that
	\begin{equation}
		\lvert \hat{u}(\xi) \rvert \leq C \langle \xi \rangle^M
	\end{equation}
	where $M$ is the order of $u$ and $\langle \xi \rangle :=(1+\lvert \xi \rvert^2)^{\frac{1}{2}}$, see \cite{FJ99}. Fix $\Gamma$ an open conic neighborhood of $\xi \in \mathbb{R}^d \setminus\{0\}$. Given $\kappa$ as per Definition \ref{Def: Kernel for Besov}, $\lambda \in (0,1)$ and $y \in \operatorname{supp}(u)$, it holds
	\[
	\bigg \lvert \int_\Gamma \hat{u}(\eta) e^{iy\cdot \eta} \check{\kappa}(\lambda \eta) d\eta \bigg\rvert \leq \int_\Gamma \lvert \hat{u}(\eta) \rvert \lvert \check{\kappa}(\lambda \eta) \rvert d\eta \leq C \int_\Gamma \langle \eta \rangle^M \lvert \check{\kappa}(\lambda \eta) \lvert d\eta \approx  \lambda^{-M-d} \int_\Gamma  \lvert \eta \rvert^{M} \lvert \check{\kappa}(\eta) \rvert d\eta \lesssim \lambda^{-M-d},
	\]
where, with reference to Equation \eqref{eq:wavefront set besov} and \eqref{eq:wavefront set besov II}, we have implicitly chosen $\phi\in\mathcal{D}(\mathbb{R}^d)$ such that $\phi=1$ on $\textrm{supp}(u)$.	
	As a result, we get $\mathrm{WF}^\alpha(u)=\emptyset$ if $\alpha\leq -d-M$. 
\end{example}

\begin{example}\label{example: sqrt of x}
	Let $u \colon \mathbb{R}^2 \to \mathbb{R}$ such that $u(x_1,x_2)=(x_1^2+x_2^2)^{\frac{1}{4}}$. We recall that $\hat{u}(\xi_1,\xi_2) = (\xi_1^2+\xi_2^2)^{-\frac{5}{4}}$, which should be interpreted as the integral kernel of an element lying in $\mathcal{S}^\prime(\mathbb{R}^2$). Since $\mathrm{singsupp}(u)=\{(0,0)\}$, we consider $(0,0,\xi_1,\xi_2)$ such that $(\xi_1,\xi_2)\neq(0,0)$. Given $\phi \in \mathcal{D}(\mathbb{R}^2)$ with $\phi(0,0) = 1$ and an open  conic neighborhood $\Gamma$ of $(\xi_1,\xi_2)$, we can still use the rationale followed in Example \eqref{Ex: BWF of delta} studying Equation \eqref{eq:wavefront set besov II} with $y=(0,0)$. It reads
	\begin{align}
		&\bigg\lvert \int_{\Gamma} \widehat{\phi u}(\eta_1,\eta_2) \check{\kappa}(\lambda\eta_1,\lambda\eta_2) d\eta_1 d\eta_2 \bigg\rvert = \bigg\lvert \int_{\Gamma} (\eta_1^2+\eta_2^2)^{-\frac{5}{4}}  \check{\kappa}(\lambda\eta_1,\lambda\eta_2)  d\eta_1 d\eta_2 \bigg\rvert=  \nonumber \\
		&\overset{(\lambda \eta_1,\lambda \eta_2) \mapsto (\eta_1,\eta_2)}{=} \int_\Gamma \lambda^{\frac{1}{2}}(\eta_1^2+\eta_2^2)^{-\frac{5}{4}} \lvert \check{\kappa}(\eta_1,\eta_2) \rvert  d\eta_1d\eta_2 = C_{\check{k}} \lambda^{\frac{1}{2}}, \nonumber
	\end{align}
where no singularity at the origin occurs since $\kappa$ is chosen in agreement with Definition \ref{Def: Kernel for Besov}. This entails that
\begin{equation}
\left\{
\begin{array}{ll}
\mathrm{WF}^\alpha(u)=\emptyset & \alpha\leq\frac{1}{2}\\
\mathrm{WF}^\alpha(u)=\{(0,0,\xi_1,\xi_2)\;|\;(\xi_1,\xi_2)\neq(0,0)\} & \alpha>\frac{1}{2}	
\end{array}
\right.
\end{equation}
\end{example}

\subsection{Pseudodifferential Characterization}\label{Sec: PsiDOWF}
The aim of this section is to give a second, albeit equivalent, characterization of the Besov wavefront set of a distribution by means of pseudodifferential operators. This is in spirit very much akin to the one outlined in \cite{Grigis} for the smooth wavefront set and it is especially useful in discussing operations between distributions with a prescribed Besov wavefront set, see Section \ref{Sec: Operations on BWF}. In the following, we shall make use of the notions introduced in Definition \ref{Def: Symbols and PsiDO} and \ref{Def: Properly Supported PsiDO}.
\begin{proposition}\label{prop: characterization via pseudodifferential operators}
	Let $\alpha \in \mathbb{R}$ and $u \in \mathcal{D}^\prime(\mathbb{R}^d)$. Then
	\begin{equation}
		\mathrm{WF}^\alpha(u) = \bigcap_{\substack{A \in \Psi^0(\mathbb{R}^d),\\ Au \in B_{\infty,\infty}^{\alpha,\mathrm{loc}}(\mathbb{R}^d)}} \mathrm{Char}(A),
	\end{equation}
	where the intersection is taken only over properly supported pseudodifferential operators.
\end{proposition}

\begin{proof}
	Suppose that $(x_0,\xi_0) \not \in \mathrm{WF}^\alpha(u)$. By Definition \ref{def:besov wavefront set}, there exist $\phi \in \mathcal{D}(\mathbb{R}^d)$ with $\phi(x_0)\neq 0$ and $\Gamma$ , a conic neighbourhood of $\xi_0$, such that for any compact set $K \subset \mathbb{R}^d$
	\[
	\bigg\lvert \int_{\Gamma} \widehat{\phi u}(\xi) \check{\kappa}(\lambda \xi) e^{ix\cdot \xi} d\xi \bigg\rvert \lesssim \lambda^\alpha \quad \forall x \in K, \forall \lambda \in (0,1],
	\]                       
	where $\kappa \in \mathscr{B}_{\alpha}$. Calling $\mathbb{I}_\Gamma(\xi)$ the characteristic function on $\Gamma$, it descends that
	\begin{equation}\label{eq: besov regularity of u}
		\mathcal{F}^{-1}\bigg[ \mathbb{I}_\Gamma(\xi) \widehat{\phi u} \bigg] \in B^{\alpha, \mathrm{loc}}_{\infty,\infty}(\mathbb{R}^d).
	\end{equation}
	 Set $\chi \in C^\infty(\mathbb{R}^n)$ to be such that $\chi(\xi)=0$ if $\lvert \xi \rvert \leq a$ and $\chi(\xi) = 1$ if $\lvert \xi \rvert \geq 2a$ where $a$ is a non vanishing constant chosen so that $\chi(\xi_0)\neq 0$. In addition choose $\psi \in C^{\infty}(\mathbb{S}^{n-1})$ such that $\mathrm{supp}(\psi) \subset B_{\varepsilon} (\xi_0/\lvert \xi_0 \rvert)\subset\Gamma$, $\varepsilon>0$ and $\psi(\xi_0/\lvert\xi_0\rvert)\neq 0$. Consequently we can introduce $A := \mathrm{Op}(a)\in\Psi^0(\mathbb{R}^d)$, where
	 \begin{equation}
	 	a(x,y,\xi) = \phi(x) \psi\bigg(\frac{\xi}{\lvert \xi \rvert}\bigg) \chi(\xi) \phi(y) \in S^0(\mathbb{R}^d \times \mathbb{R}^d;\mathbb{R}^d).
	 \end{equation}
Observe that, following standard arguments, $A$ is by construction properly supported and elliptic at $(x_0,\xi_0)$. To conclude it suffices to notice that, combining Equation \eqref{eq: besov regularity of u} and Theorem \ref{th: boundedness on besov spaces}, we can conclude that $Au\in B_{\infty,\infty}^{\alpha,\mathrm{loc}}(\mathbb{R}^d)$.

Conversely, let $(x_0,\xi_0) \not \in \bigcap_{\substack{A \in \Psi^0(\mathbb{R}^d)\\ Au \in B_{\infty,\infty}^{\alpha,\mathrm{loc}}(\mathbb{R}^d)}} \mathrm{Char}(A)$. 
Hence, taking into account Definition \ref{def: elliptic set}, there exists $B \in \Psi^0$, elliptic at $(x_0,\xi_0)$, such that $Bu \in B_{\infty,\infty}^{\alpha,\mathrm{loc}}(\mathbb{R}^d)$. 
Consider once more $\phi$, $\psi$ and $\chi$ as in the previous part of the proof, so that
\[
WF^\prime(A) \subset \mathrm{Ell}(B)
\]
where $A:=Op_R(\psi(\xi/\lvert \xi \rvert) \chi(\xi) \phi(y))$ and where $WF^\prime$ is as per Definition \ref{Def: Operator WF}. 
We claim that $Au \in B_{\infty,\infty}^{\alpha,\mathrm{loc}}(\mathbb{R}^d)$. 
In view of Proposition \ref{prop: parametrix}, there exists a microlocal parametrix $Q \in \Psi^0(\mathbb{R}^d)$ of $B$ such that $QB = I-R$ with $R \in \Psi^{-1}(\mathbb{R}^d)$ and $WF^\prime(R) \cap WF^\prime(A) = \emptyset$. Thus,
\[
Au = A(QB+R) u = (AQ)(Bu) + ARu,
\]
where $ARu \in C^\infty(\mathbb{R}^d)$. 
Given $\rho \in \mathcal{D}(\mathbb{R}^d)$ such that $\rho = 1$ on $\mathrm{supp}(\phi)$, it descends
\[
(AQ)(Bu) = (AQ)(\rho Bu) + (AQ)((1-\rho)Bu).
\]
Since $1-\rho = 0$ on $\mathrm{supp}(\phi)$, then $(AQ)((1-\rho)Bu) = 0$. 
At the same time $(AQ)(\rho Bu) \in B^\alpha _{\infty,\infty}(\mathbb{R}^d)$ on account of Theorem \ref{th: boundedness on besov spaces}. 
This entails that $Au \in B^{\alpha,\mathrm{loc}}_{\infty,\infty}(\mathbb{R}^d)$. Hence, given $\kappa \in \mathscr{B}_\alpha$, see Definition \ref{Def: Kernel for Besov}, it holds 
\begin{equation}\label{eq: besov regularity psido}
	\bigg \lvert \int_{\mathrm{Ell}\big( \psi(D/\lvert D\rvert) \chi(D)\big)} \psi\bigg(\frac{\xi}{\lvert \xi \rvert} \bigg ) \chi(\xi) \widehat{\phi u}(\xi) \check{\kappa}(\lambda \xi) e^{i x \cdot \xi}d\xi\bigg \rvert \lesssim \lambda^\alpha, \quad \forall \lambda \in (0,1], \quad \forall x \in K.
\end{equation}
On account of Remark \ref{remark: equivalent def elliptic set}, there exists a symbol $p \in S^0(\mathbb{R}^n;\mathbb{R}^n)$ such that
\[
r(\xi) : = 1-\psi\bigg(\frac{\xi}{\lvert \xi \rvert} \bigg ) \chi(\xi) p(\xi)  \in S^{-1}
\]
for any $\xi \in \mathrm{Ell}\big( \psi(D/\lvert D\rvert) \chi(D)\big)$. It descends
\begin{gather}
	\bigg \lvert \int_{\mathrm{Ell}\big( \psi(D/\lvert D\rvert) \chi(D)\big)} \widehat{\phi u}(\xi) \check{\kappa}(\lambda \xi) e^{i x \cdot \xi}d\xi\bigg \rvert =\\
	= \bigg \lvert \int_{\mathrm{Ell}\big( \psi(D/\lvert D\rvert) \chi(D)\big)} \bigg (\psi\bigg(\frac{\xi}{\lvert \xi \rvert} \bigg ) \chi(\xi) p(\xi) + r(\xi) \bigg) \widehat{\phi u}(\xi) \check{\kappa}(\lambda \xi) e^{i x \cdot \xi}d\xi\bigg \rvert \nonumber \\
	\leq \bigg \lvert \int_{\mathrm{Ell}\big( \psi(D/\lvert D\rvert) \chi(D)\big)} \psi\bigg(\frac{\xi}{\lvert \xi \rvert} \bigg ) \chi(\xi) p(\xi) \widehat{\phi u}(\xi) \check{\kappa}(\lambda \xi) e^{i x \cdot \xi}d\xi\bigg \rvert \nonumber + \bigg \lvert \int_{\mathrm{Ell}\big( \psi(D/\lvert D\rvert) \chi(D)\big)} r(\xi) \widehat{\phi u}(\xi) \check{\kappa}(\lambda \xi) e^{i x \cdot \xi}d\xi\bigg \rvert  \nonumber \\
	= \underbrace{\bigg\lvert \bigg\langle p(D)\psi\bigg(\frac{D}{\lvert D \rvert}\bigg) \chi(D)(\phi u), \kappa^\lambda_x \bigg\rangle  \bigg\rvert}_{\lvert A \rvert} + \bigg \lvert \int_{\mathrm{Ell}\big( \psi(D/\lvert D\rvert) \chi(D)\big)} r(\xi) \widehat{\phi u}(\xi) \check{\kappa}(\lambda \xi) e^{i x \cdot \xi}d\xi\bigg \rvert, \nonumber 
\end{gather}
for any $x \in K$ and $\lambda \in (0,1]$. 
On the one hand, as a result of Theorem \ref{th: boundedness on besov spaces} and Equation \eqref{eq: besov regularity psido}, it holds that
\[
\lvert A \rvert \lesssim \lambda^\alpha.
\]
On the other hand, 
\begin{align}
\lvert B \rvert \leq \bigg\lvert \bigg \langle r(D)p(D)\psi\bigg(\frac{D}{\lvert D \rvert}\bigg) \chi(D)(\phi u), \kappa^\lambda_x \bigg \rangle \bigg\rvert &+ \bigg \lvert \int_{\mathrm{Ell}\big( \psi(D/\lvert D\rvert) \chi(D)\big)} r^2(\xi) \widehat{\phi u}(\xi) \check{\kappa}(\lambda \xi) e^{i x \cdot \xi}d\xi\bigg \rvert \lesssim  \nonumber \\
&\lambda^{\alpha + 1} + \bigg \lvert \int_{\mathrm{Ell}\big( \psi(D/\lvert D\rvert) \chi(D)\big)} r^2(\xi) \widehat{\phi u}(\xi) \check{\kappa}(\lambda \xi) e^{i x \cdot \xi}d\xi\bigg \rvert,
\end{align}
where we applied once more Theorem \ref{th: boundedness on besov spaces} with $r(D)\in \Psi^{-1}(\mathbb{R}^d)$ and $p(D)\psi\bigg(\frac{D}{\lvert D \rvert}\bigg) \chi(D)(\phi u) \in B_{\infty,\infty}^{\alpha,\mathrm{loc}}(\mathbb{R}^d)$. 
This concludes the proof.
\end{proof}

\begin{remark}\label{Rem: PsiDO and smooth WF}
	The content of Proposition \ref{prop: characterization via pseudodifferential operators} is an adaptation to the case in hand of the characterization of the smooth wavefront set of a distribution in terms of pseudodifferential operators, see \cite[Cor. 6.18]{Hintz}. For later convenience and to fix the notation, we recall it. Let $v\in\mathcal{D}^\prime(\mathbb{R}^d)$. It holds 
	$$WF(v)=\bigcap_{\substack{A \in \Psi^0(\mathbb{R}^d)\\ Av \in C^\infty(\mathbb{R}^d)}}\mathrm{Char}(A),$$
	where $\mathrm{Char}(A)$ is the characteristic set of $A$ introduced in Definition \ref{def: elliptic set}.
\end{remark}

\noindent We prove a proposition aimed at stating another useful characterization of the Besov wavefront set of a distribution.

\begin{proposition}\label{prop:property besov wavefront}
	Let $u \in \mathcal{D}^\prime(\mathbb{R}^d)$. It holds that
	\begin{equation}
		(x,\xi) \in \mathrm{WF}^\alpha(u) \iff (x,\xi) \in WF(u-v)\quad \forall v \in B_{\infty,\infty}^{\alpha,\mathrm{loc}}(\mathbb{R}^d),
	\end{equation}
where $WF$ stands for the (smooth) wavefront set.
\end{proposition}
\begin{proof}
	Suppose $(x,\xi) \in\mathrm{WF}^\alpha(u)$. On account of Remark \ref{Rem: PsiDO and smooth WF}, given $v\in B_{\infty,\infty}^{\alpha,\mathrm{loc}}(\mathbb{R}^d)$ we consider $A \in \Psi^0(\mathbb{R}^d)$ such that $A(u-v) \in C^\infty(\mathbb{R}^d)$. This entails that $Au \in B_{\infty,\infty}^{\alpha,\mathrm{loc}}(\mathbb{R}^d)$. Yet, since $(x,\xi) \in \mathrm{WF}^\alpha(u)$, Proposition \ref{prop: characterization via pseudodifferential operators} entails that $(x,\xi)\in\mathrm{Char}(A)$.\\
	Conversely, let $(x,\xi) \not \in \mathrm{WF}^\alpha(u)$. By Definition \ref{def:besov wavefront set}, there exist $\phi \in \mathcal{D}(\mathbb{R}^d)$ normalized so that $\phi(x)=1$ and an open conic neighborhood $\Gamma$ of $\xi$ such that Equation \eqref{eq:wavefront set besov} is satisfied. Let $v \in B_{\infty,\infty}^{\alpha,\mathrm{loc}}(\mathbb{R}^d)$ be such that
	\begin{equation}
		\hat{v}(\eta) =
		\begin{cases}
			\widehat{\phi u}(\eta) \quad \text{if } \eta \in \Gamma,\\
			0 \quad \text{otherwise}
		\end{cases}
	\end{equation}
	Then $\hat{\theta} = \widehat{\phi u} - \hat{v}$ vanishes on $\Gamma$ and therefore $(x,\xi) \not \in WF(\theta)$. Consider $\chi \in \mathcal{D}(\mathbb{R}^d)$ such that $\chi \phi = 1$ in a neighbourhood of $x\in\mathbb{R}^d$. Then $\chi v \in B_{\infty,\infty}^\alpha(\mathbb{R}^d)$ and $(x,\xi) \not \in WF(\chi \theta)$. After observing that $u-\chi v = (1-\chi\phi) u + \chi \theta$, we conclude $(x,\xi) \not \in WF(u-\chi v)$ exploiting that $(1-\chi\phi)u$ vanishes in a neighbourhood of $x$. Since $\chi=1$ at $x$, we can conclude that $(x,\xi)\notin WF(u-v)$.
\end{proof}

\noindent We can now establish a relation between the Besov wavefront set and the smooth counterpart, see Remark \ref{Rem: PsiDO and smooth WF}. The second part of the proof of the following corollary is inspired by a similar one, valid in the context of the Sobolev wavefront set \cite[Prop. 6.32]{Hintz}.

\begin{corollary}
	Let $u\in\mathcal{D}^\prime(\mathbb{R}^d)$. It holds that
	\begin{equation}\label{Eq: BWF vs WF}
		WF(u)=\overline{\bigcup_{\alpha \in \mathbb{R}}\mathrm{WF}^\alpha(u)}.
	\end{equation}
\end{corollary}

\begin{proof}
	Assume $(x,\xi)\in\mathrm{WF}^\alpha(u)$ for any $\alpha\in\mathbb{R}$. Using Proposition \ref{prop:property besov wavefront}, we can choose $v\in C^\infty_0(\mathbb{R}^d)\subset B^{\alpha,\mathrm{loc}}_{\infty,\infty}(\mathbb{R}^d)$ concluding that $(x,\xi)\in WF(u-v)=WF(u)$. Hence 
	$\bigcup_{\alpha \in \mathbb{R}}\mathrm{WF}^\alpha(u)\subseteq WF(u)$. Taking the closure and recalling that $WF(u)$ is per construction a closed set, it descends $\overline{\bigcup_{\alpha \in \mathbb{R}}\mathrm{WF}^\alpha(u)}\subseteq WF(u)$. 
	
	To prove the other inclusion, assume $(x,\xi)\notin WF^\alpha(u)$ for all $\alpha\in\mathbb{R}$. Hence there must exist a conic, open set $\Gamma\subseteq\mathbb{R}^d\times\mathbb{R}^d\setminus\{0\}$ such that $(x,\xi)\in\Gamma$ and $\Gamma\cap\mathrm{WF}^\alpha(u)=\emptyset$ for all $\alpha\in\mathbb{R}$. We can thus choose $A\in\Psi^0(\mathbb{R}^d)$ to be properly supported, elliptic at $(x,\xi)$ and such that $WF^\prime(A)\subset\Gamma$ and $Au\in B^{\alpha,\mathrm{loc}}_{\infty,\infty}(\mathbb{R}^d)$ for all $\alpha\in\mathbb{R}$. This entails that $Au\in C^\infty(\mathbb{R}^d)$. It descends that, since $(x,\xi)\notin\mathrm{Char}(A)$, then $(x,\xi)\notin WF(u)$.
\end{proof}

\section{Structural Properties}\label{Sec: Operations on BWF}

In this section we discuss the main structural properties of distributions with a prescribed Besov wavefront set as per Definition \ref{def:besov wavefront set} and Proposition \ref{prop: characterization via pseudodifferential operators}, including notable operations.

\paragraph{Transformation Properties under Pullback --} We start by investigating the interplay between Definition \ref{def:besov wavefront set} and the pull-back of a distribution. In the following we enjoy the analysis outlined in Section \ref{Sec: Localization of a PsiDO}. 

\begin{remark}\label{Rem: Besov WF on open subset}
	In Definition \ref{def:besov wavefront set} as well as in Proposition \ref{prop: characterization via pseudodifferential operators} we have always assumed implicitly that the underlying distribution is globally defined, {\it i.e.} $u\in\mathcal{D}^\prime(\mathbb{R}^d)$. Yet, mutatis mutandis, the whole construction and the results obtained so far can be slavishly adapted to distributions $v\in\mathcal{D}^\prime(\Omega)$, $\Omega\subseteq\mathbb{R}^d$.
\end{remark}

\begin{theorem}[Pull-back - I]\label{th: pullback theorem 2}
	Let $\Omega\subseteq \mathbb{R}^d$, $\Omega^\prime\subseteq \mathbb{R}^m$ be open sets and let $f\in C^\infty(\Omega;\Omega^\prime)$ be an embedding. Moreover let
	\begin{equation}\label{Eq: normal directions}
	N_f:=\{(f(x),\xi) \in \Omega^\prime \times \mathbb{R}^m: {}^t df(x) \xi = 0\},	
	\end{equation}
	be the set of normals of $f$. For any $u \in \mathcal{D}^\prime(\Omega^\prime)$ such that there exists $\alpha > 0$ so that
	\begin{equation}\label{eq:condition pull back 2}
		N_f \cap \mathrm{WF}^\alpha(u) = \emptyset,
	\end{equation}
	there exists $f^\ast u \in \mathcal{D}^\prime(\Omega)$. In addition 
	\begin{equation}\label{eq: inclusion pull back 2}
		\mathrm{WF}^\alpha(f^\ast u) \subseteq f^\ast \mathrm{WF}^\alpha(u),
	\end{equation}
	for every $u \in \mathcal{D}^\prime(\Omega^\prime)$ abiding to Equation \eqref{eq:condition pull back 2}, where 
	\begin{equation}\label{eq: def pull back by diffeo}
		f^\ast \mathrm{WF}^\alpha(u) := \{(x,{}^t df(x)\eta):(f(x),\eta) \in \mathrm{WF}^\alpha(u)\}.
	\end{equation}
\end{theorem}
\begin{proof}
	As a consequence of Proposition \ref{prop:property besov wavefront}, Equation \eqref{eq:condition pull back 2} is equivalent to
	\[
	N_f \cap  WF(u-v) = \emptyset, \quad \forall v \in B_{\infty,\infty}^{\alpha,\mathrm{loc}}(\Omega^\prime).
	\]
	Then, there exists the pullback $f^\ast(u-v) \in \mathcal{D}^\prime(\Omega)$. Taking into account that $B_{\infty,\infty}^{\alpha,\mathrm{loc}}(\Omega^\prime) \subset C^0(\Omega^\prime)$ for $\alpha > 0$, we have that $f^\ast v = v \circ f$. Thus,
	\[
	f^\ast u = f^\ast(u-v) + f^\ast v
	\]
	identifies an element lying in $\mathcal{D}^\prime(\Omega)$. Focusing on Equation \eqref{eq: inclusion pull back 2}, let $(x,{}^t df(x)\eta) \not \in f^\ast \mathrm{WF}^\alpha(u)$. It implies $(f(x),\eta) \not \in \mathrm{WF}^\alpha(u)$. By Proposition \ref{prop: characterization via pseudodifferential operators}, there exists $A \in \Psi^0(\Omega^\prime)$, elliptic in $(f(x),\eta)$, such that $Au \in B^{\alpha,\mathrm{loc}}_{\infty,\infty}(\Omega^\prime)$. Bearing in mind that $f$ is a diffeomorphism on $f[\Omega]$, 
	\[
	A_f(f^\ast u) = (Au) \circ f,
	\]
	identifies a pseudodifferential operator of order $0$ per Theorem \ref{th: psido coordinate change}. Since $(Au) \circ f \in B_{\infty,\infty}^{\alpha,\mathrm{loc}}(\mathbb{R}^d)$, Theorem \ref{th: psido coordinate change} entails
	\[
	\sigma_0(A_f)(x,{}^t df(x)\eta) = \sigma_0(A)(f(x), \eta) \neq 0.
	\]
	This proves $(x,{}^t df(x)\eta) \not \in \mathrm{WF}^\alpha(f^\ast u)$.
\end{proof}

To conclude this first part of the section, we shall prove that Besov wavefront set is invariant under the action of diffeomorphisms. 

\begin{theorem}[Pull-back - II]\label{Thm: BWF and Diffeo}
	Let $\Omega,\Omega^\prime\subseteq \mathbb{R}^d$ be two open subsets and let $f \colon \Omega\to\Omega^\prime$ be a diffeomorphism. Then, given $u \in \mathcal{D}^\prime(\Omega^\prime)$, for any $\alpha\in\mathbb{R}$, it holds
	\[
	\mathrm{WF}^\alpha(f^\ast u) = f^\ast \mathrm{WF}^\alpha(u).
	\]
\end{theorem}
\begin{proof}
	We prove the inclusion $\mathrm{WF}^\alpha(f^\ast u) \subseteq f^\ast \mathrm{WF}^\alpha(u)$. 
	Let $(x,\xi) \not \in f^\ast \mathrm{WF}^\alpha(u)$, {\it i.e.}, $(f(x),({}^t df(x))^{-1} \xi)\not \in \mathrm{WF}^\alpha(u)$. 
	Thus, there exists $A \in \Psi^0(\Omega^\prime)$, elliptic at $(f(x),({}^t df(x))^{-1} \xi)$, such that $Au \in B_{\infty,\infty}^{\alpha,\mathrm{loc}}(\Omega^\prime)$. If one introduces $A_f \in \Psi^0(\Omega)$ such that
	\[
	A_f(f^\ast u) = f^\ast (Au),
	\]
	Equation \eqref{eq: coordinate change principal symbol} entails that
	\[
	\sigma_0(A_f)(x,\xi) = \sigma_0(A)(f(x), ({}^t df(x))^{-1} \xi) \neq 0,
	\]
	that is $A_f$ is elliptic at $(f(x), ({}^t df(x))^{-1} \xi)$. 
	To conclude, we need to prove that $A_f(f^\ast u) \in B_{\infty,\infty}^{\alpha,\mathrm{loc}}(\Omega)$.
	For any but fixed $\phi \in C_0^\infty(\Omega)$, it holds
	\begin{align}
		\lvert \langle \phi A_f(f^\ast u), \kappa^\lambda_z  \rangle \rvert = \lvert \langle  f^\ast(Au), \phi \kappa^\lambda_z  \rangle \rvert = \lvert \langle  Au, (f_\ast\phi) (f_\ast\kappa)^\lambda_{f(z)} \lvert \det(df^{-1}) \rvert  \rangle \rvert \lesssim \lvert \langle  Au, (f_\ast\phi) (f_\ast\kappa)^\lambda_{f(z)} \rangle \rvert  \lesssim \lambda^\alpha,\nonumber
	\end{align}
	for any $z \in \Omega$, $\lambda \in (0,1]$ and $\kappa \in \mathscr{B}_\alpha$. 
	An analogous estimate yields
	\[
	\lvert \langle \phi A_f(f^\ast u), \underline{\kappa}_z  \rangle \rvert \lesssim 1,
	\]
	for any $z \in\Omega$ and $\underline{\kappa} \in \mathcal{D}(B(0,1))$ such that $\check{\underline{\kappa}}(0)\neq 0$. 
	This proves that $(x,\xi) \not \in \mathrm{WF}^\alpha(f^\ast u)$.\\
	Conversely, let $(x,\xi) \not \in \mathrm{WF}^\alpha(f^\ast u)$. 
	Then, there exists $\tilde{A} \in \Psi^0(\Omega)$, elliptic in $(x,\xi)$, such that $\tilde{A}(f^\ast u) \in B^{\alpha,\mathrm{loc}}_{\infty,\infty}(\Omega)$. 
	Let  $A \in \Psi^0(\Omega^\prime)$ be such that
	\[
	Au = (f^{-1})^\ast(\tilde{A}(f^\ast u)) .
	\]
	Still on account of Theorem \ref{th: psido coordinate change}, it holds that
	\[
	\sigma_0(A)(f(x),({}^t df(x))^{-1}\xi) = \sigma_0(\tilde{A})(x,\xi) \neq 0.
	\]
	As a consequence, $A$ is elliptic at $(f(x),({}^t df(x))^{-1}\xi)$. Reasoning as in the first part of the proof, it turns out that $Au \in B^{\alpha,\mathrm{loc}}_{\infty,\infty}(\Omega^\prime)$. This entails that $(x,\xi) \not \in f^\ast \mathrm{WF}^\alpha(u)$.
\end{proof}

\begin{remark}\label{Rem: Extension of BWF to manifolds}
	Theorem \ref{Thm: BWF and Diffeo} is especially noteworthy since it is the building block to extend the notion of Besov wavefront set to distributions supported on any arbitrary smooth manifold $M$, following the same rationale used when working with the smooth counterpart.
	On a similar note, we observe that for the sake of simplicity of the presentation, we decided to stick to individuating a point of $\mathrm{WF}^\alpha(u)$, $u\in\mathcal{D}^\prime(\mathbb{R}^d)$, as an element of $\mathbb{R}^d\times\mathbb{R}^d\setminus\{0\}$. Yet, from a geometrical viewpoint each element of $\mathrm{WF}^\alpha(u)$ should be better read as lying in the cotangent bundle $\mathrm{T}^*\mathbb{R}^d\setminus\{0\}$. For the sake of conciseness, we shall not dwell into further details which are left to the reader.
\end{remark}

\paragraph{Microlocal Properties of $\Psi$DOs --} Our next task is the study of the interplay between pseudodifferential operators and distributions at the level of wavefront set. To this end we recall a notable result, valid in the smooth setting, see \cite[Prop. 6.27]{Hintz}, namely, if $A\in\Psi^m(\mathbb{R}^d)$ and $u\in\mathcal{D}^\prime(\mathbb{R}^d)$, then $A$ is {\em microlocal}:
\begin{equation}\label{eq: pseudolocality smooth wavefront set}
	WF(Au) \subseteq WF^\prime(A) \cap WF(u),
\end{equation}
where $WF^\prime$ stands for the operator wavefront set as per Definition \ref{Def: Operator WF}. At the level of Besov wavefront set the counterpart of this statement is the following proposition.

\begin{proposition}\label{prop: pseudolocality}
	Let $A \in \Psi^m(\mathbb{R}^d)$, $u \in \mathcal{D}^\prime(\mathbb{R}^d)$ and $\alpha \in \mathbb{R}$. Then 
	\begin{equation}\label{eq: pseudolocality}
		\mathrm{WF}^{\alpha - m}(Au) \subseteq WF^\prime(A) \cap \mathrm{WF}^\alpha(u).
	\end{equation}
\end{proposition}
\begin{proof}
	Suppose that $(x_0,\xi_0) \not \in WF^\prime(A)$. As a consequence of Proposition \ref{prop: parametrix} there exists $B \in \Psi^0(\mathbb{R}^d)$, elliptic at $(x_0,\xi_0)$. 
	In addition, we find $B$ such that $WF^\prime(A) \cap WF^\prime(B) = \emptyset$. 
	Proposition \ref{Prop: WF Properties} entails that $BA \in \Psi^{-\infty}(\mathbb{R}^d)$, which implies in turn that $B(Au)\in C^\infty(\mathbb{R}^d)\subseteq B_{\infty,\infty}^{\alpha-m,\mathrm{loc}}(\mathbb{R}^d)$. Proposition \ref{prop: characterization via pseudodifferential operators} yields that $(x_0,\xi_0)\notin\mathrm{WF}^{\alpha-m}(Au)$.
	
	Conversely, suppose $(x_0,\xi_0) \not \in \mathrm{WF}^\alpha(u)$. Then, still in view of Proposition \ref{prop: characterization via pseudodifferential operators}, there exists $\tilde{A} \in \Psi^0(\mathbb{R}^d)$, elliptic at $(x_0,\xi_0)$, such that $\tilde{A}u \in B^{\alpha,\mathrm{loc}}_{\infty,\infty}(\mathbb{R}^d)$. Take $B \in \Psi^0(\mathbb{R}^d)$ elliptic at $(x_0,\xi_0)$ with $WF^\prime(B) \subseteq \mathrm{Ell}(\tilde{A})$. On account of Proposition \ref{prop: parametrix}, there exists a parametrix $Q \in \Psi^0(\mathbb{R}^d)$ of $\tilde{A}$, that is, $Q\tilde{A} = I-R$ with $R \in \Psi^0(\mathbb{R}^d)$ and $WF^\prime(R) \cap WF^\prime(B) = \emptyset$. Therefore, 
	\[
	B(Au) = BA(Q\tilde{A}+R)u = BAQ(\tilde{A}u) + (BAR)u.
	\]
	Since $BAR \in \Psi^{-\infty}(\mathbb{R}^d)$, then $(BAR)u \in C^\infty(\mathbb{R}^d)$. At the same time $BAQ(\tilde{A}u) \in B^{\alpha - m,\mathrm{loc}}_{\infty,\infty}(\mathbb{R}^d)$ because $\tilde{A}u \in B^{\alpha,\mathrm{loc}}_{\infty,\infty}(\mathbb{R}^d)$ and $BAQ \in \Psi^{m}(\mathbb{R}^d)$. Yet, since $(x_0,\xi_0)\notin \mathrm{Char}(B)$, it descends $(x_0,\xi_0)\notin\mathrm{WF}^{\alpha-m}(Au)$. This concludes the proof.
\end{proof}

The second result we present in this section provides a sort of inverse result, with respect to the previous one, which is more relevant from a PDE viewpoint.

\begin{proposition}\label{prop: char with psdo}
	Let $u \in \mathcal{D}^\prime(\mathbb{R}^d)$, $A \in \Psi^m(\mathbb{R}^d)$ and $m,\alpha \in \mathbb{R}$. Then
	\begin{equation}\label{eq: char with psido}
		\mathrm{WF}^\alpha(u) \subseteq \mathrm{Char}(A) \cup  \mathrm{WF}^{\alpha - m}(A u).
	\end{equation}
\end{proposition}

\begin{proof}
	Let $(x_0,\xi_0) \not \in \mathrm{Char}(A) \cup  \mathrm{WF}^{\alpha - m}(A u)$. Thus there exists $B \in \Psi^0(\mathbb{R}^d)$, elliptic at $(x_0,\xi_0)$, such that $B(Au) \in B^{\alpha-m,\mathrm{loc}}_{\infty,\infty}(\mathbb{R}^d)$. Let $K$ be any properly supported pseudodifferential operator lying in $\Psi^{-m}(\mathbb{R}^n)$, which can be chosen without loss of generality to be elliptic at $(x_0,\xi_0)$. It descends $(KBA) u \in B^{\alpha,\mathrm{loc}}_{\infty,\infty}(\mathbb{R}^d)$. Since $(x_0,\xi_0)\notin\mathrm{Char}(KBA)$, it descends that $(x_0,\xi_0)\notin\mathrm{WF}^\alpha(u)$. 
\end{proof}

\begin{corollary}[Elliptic Regularity]
	Let $u \in \mathcal{D}^\prime(\mathbb{R}^d)$, $m,\alpha \in \mathbb{R}$ and let $A\in\Psi^m(\mathbb{R}^d)$ be elliptic. Then
	$$\mathrm{WF}^\alpha(u) = \mathrm{WF}^{\alpha - m}(A u).$$
\end{corollary}

\begin{proof}
	Since $A$ is an elliptic pseudodifferential operator $\mathrm{Char}(A)=\emptyset$ and, in view of Definition \ref{Def: Operator WF}, $WF^\prime(A)=\emptyset$. The statement is thus a direct consequence of Propositions \ref{prop: pseudolocality} and \ref{prop: char with psdo}.
\end{proof}

\paragraph{Product of Distributions --} In the following we investigate the formulation of a version of H\"ormander's criterion for the product of distributions, tied to the Besov wavefront set. In the spirit of \cite{Hor03}, we rely on two ingredients. The first has already been discussed in Theorem \ref{th: pullback theorem 2}, while the second one concerns the tensor product of two distributions. In particular we wish to establish an estimate on the singular behaviour of $u\otimes v$ for given $u,v\in\mathcal{D}^\prime(\mathbb{R}^d)$. This can be read as a direct adaptation to this context of \cite[Prop. B.5]{Junker:2001gx} which is based in turn on \cite[Lemma 11.6.3]{Hor97}. For this reason we shall omit the proof.

\begin{proposition}[Tensor product]\label{Prop: besov tensor product theorem 2}
	Let $\Omega\subseteq \mathbb{R}^d$ and $\Omega^\prime\subseteq \mathbb{R}^m$ be two open sets. If $u\in \mathcal{D}^\prime(\Omega)$ and $v \in \mathcal{D}^\prime(\Omega^\prime)$, then the following two inclusions hold true:
\begin{equation}\label{Eq: WF-Tensor Product}
\mathrm{WF}^{\alpha+\beta}(u\otimes v)\subseteq \mathrm{WF}_0^\alpha(u)\times WF(v)\cup WF(u)\times\mathrm{WF}^\beta_0(v),
\end{equation}	
and, calling $\gamma:=\min\{\alpha,\beta,\alpha+\beta\}$,
\begin{equation}\label{Eq: WF-Tensor Product_improved}
	\mathrm{WF}^\gamma(u\otimes v)\subseteq \mathrm{WF}^\alpha(u)\times WF_0(v)\cup WF_0(u)\times\mathrm{WF}^\beta(v),
\end{equation}
where we adopted the notation $WF_0(u):=WF(u)\cup\left(\textrm{supp}(u)\times \{0\}\right)$ and similarly $\mathrm{WF}^\alpha_0(u):=\mathrm{WF}^\alpha(u)\cup\left(\textrm{supp}(u)\times \{0\}\right)$.
\end{proposition}

\noindent At last, we are in a position to prove a counterpart of H\"ormander's criterion for the product of two distributions within the framework of the Besov wavefront set.

\begin{theorem}\label{th:product theorem 2}
	Let $u,v \in \mathcal{D}^\prime(\Omega)$ where $\Omega\subseteq \mathbb{R}^d$ is any open set. If $\forall (x,\xi)\in\Omega\times\mathbb{R}^d\setminus\{0\}$ there exist $\alpha,\beta \in \mathbb{R}$ with $\alpha + \beta > 0$ such that 
	\[
	(x,\xi) \not \in \mathrm{WF}^\alpha(u)\cup (-WF^\beta(v))
	\]
	then the product $uv\in\mathcal{D}^\prime(\Omega)$ can be defined by
	\[
	u  v = \Delta^\ast(u \otimes v),
	\]
	where $\Delta \colon\Omega\to\Omega\times\Omega$ is the diagonal map. Moreover, calling $\gamma:=\min\{\alpha,\beta\}$,
	\begin{equation}\label{Eq: WF estimate product}
		\mathrm{WF}^\gamma(uv) \subset \{(x,\xi+\eta):(x,\xi) \in \mathrm{WF}^\alpha(u), (x,\eta) \in WF_0(v)\;\textrm{or}\;(x,\xi)\in WF_0(u), (x,\eta) \in \mathrm{WF}^\beta(v)\}.
	\end{equation}
\end{theorem}
\begin{proof}
	Observe that, per hypothesis, there exist $\alpha,\beta \in \mathbb{R}$ with $\alpha +\beta > 0$ such that
	\[
	\mathrm{WF}^{\alpha+\beta}(u\otimes v) \cap N_{\Delta} = \emptyset,
	\]
	where $N_{\Delta}=\{(x,x,\xi,-\xi)\}$ is the set of normal directions of the diagonal map as defined in Equation \eqref{Eq: normal directions}. Hence, on account of Theorem \ref{th: pullback theorem 2} combined with Proposition \ref{Prop: besov tensor product theorem 2}, Equation \eqref{Eq: WF-Tensor Product_improved} in particular, there exists $\Delta^\ast(u\otimes v) \in \mathcal{D}^\prime(\mathbb{R}^d)$ and
	\begin{gather*}
	\mathrm{WF}^\gamma(\Delta^\ast(u\otimes v)) \subset \Delta^\ast \mathrm{WF}^\gamma(u\otimes v) =\\ =\{(x,\xi+\eta):(x,\xi) \in \mathrm{WF}^\alpha(u), (x,\eta) \in WF_0(v)\;\textrm{or}\;(x,\xi)\in WF_0(u), (x,\eta) \in \mathrm{WF}^\beta(v)\}.
	\end{gather*}
This concludes the proof.
\end{proof}

\begin{remark}\label{Rem: Classical Young Theorem}
	Observe that, if we consider $u\in B^{\alpha,\textrm{loc}}_{\infty\infty}(\mathbb{R}^d)$ and $v\in B^{\beta,\textrm{loc}}_{\infty\infty}(\mathbb{R}^d)$ with $\alpha+\beta>0$, it descends that $\mathrm{WF}^\alpha(u)=\mathrm{WF}^\beta(v)=\emptyset$. Hence, on account of Theorem \ref{th:product theorem 2}, there exists $uv\in\mathcal{D}^\prime(\mathbb{R}^d)$ and $\mathrm{WF}^\gamma(uv)=\emptyset$ with $\gamma=\min\{\alpha,\beta\}$. This is nothing but the statement of the renown Young's theorem on the product of two H\"older distributions, see \cite{BCD11, DRS21}.
\end{remark}

To conclude this section, we discuss an application of Theorem \ref{th:product theorem 2} which is of relevance in many concrete scenarios. More precisely, we consider a continuous $\mathcal{K}\colon C_0^\infty(\Omega^\prime) \to \mathcal{D}^\prime(\Omega)$ with kernel $K \in \mathcal{D}^\prime(\Omega\times\Omega^\prime)$. Given $u \in\mathcal{E}^\prime(\Omega^\prime)$, we investigate the existence of $\mathcal{K}u$ and we seek to establish a bound on the associated Besov wavefront set. As a preliminary step, we need to prove two ancillary results.

\begin{corollary}\label{Cor: projection on Besov spaces}
	Let $\Omega\times\Omega^\prime\subseteq\mathbb{R}^d\times\mathbb{R}^m$ and let $v\in B^{\alpha,\mathrm{loc}}_{\infty,\infty}(\Omega\times\Omega^\prime)$, $\alpha\in\mathbb{R}$. Calling $\pi:\Omega\times\Omega^\prime\to\Omega$ the projection map on the first factor, it holds that $\pi_\ast v\in B^{\alpha,\mathrm{loc}}_{\infty,\infty}(\Omega)$, $\pi_*$ being the push-forward map.
\end{corollary}

\begin{proof}
	Without loss of generality, let us consider $v\in B^{\alpha,\mathrm{loc}}_{\infty,\infty}(\Omega\times\Omega^\prime)\cap \mathcal{E}^\prime(\Omega\times\Omega^\prime)$. We recall that 
	\[
	(\pi_\ast v)(\phi) := v(\phi \otimes 1),
	\]
	where $\phi \in\mathcal{E}(\Omega)$. Then, for any  $\kappa \in \mathscr{B}_{\lfloor\alpha\rfloor}(\Omega)$ as per Definition \ref{Def: Kernel for Besov}, $x^\prime \in \Omega$, $\lambda \in (0,1]$,
	\[
	\lvert (\pi_\ast v)(\kappa^\lambda_{x^\prime}) \rvert = \lvert v((\kappa \otimes 1)^\lambda_{(x^\prime,y^\prime)}) \rvert \lesssim \lambda^\alpha.
	\]
Observe that $\kappa \otimes 1 \in \mathscr{B}_{\lfloor\alpha\rfloor}(\Omega\times\Omega^\prime)$. At the same time, for any $\underline{\kappa} \in \mathcal{D}(B(0,1))$ with $\check{\underline{\kappa}}(0)\neq 0$, $x^\prime \in \Omega$, $\lambda \in (0,1]$, it holds true
	\[
	\lvert (\pi_\ast v)(\underline{\kappa}_{x^\prime}) \rvert = \lvert v((\underline{\kappa} \otimes 1)_{(x^\prime,y^\prime)}) \rvert \lesssim \lambda^\alpha,
	\]
which concludes the proof.
\end{proof}

\begin{proposition}\label{Prop: ancillary}
	Let $v\in\mathcal{E}^\prime(\Omega\times\Omega^\prime)$, where $\Omega\times\Omega^\prime\subseteq\mathbb{R}^d\times\mathbb{R}^m$ is an open subset. Assume that the projection map on the first factor $\pi:\Omega\times\Omega^\prime\to\Omega$ is proper on $\textrm{supp}(v)$. Then it holds that, for all $\alpha\in\mathbb{R}$
    $$\mathrm{WF}^\alpha(\pi_*v)\subset\{(x,\xi)\in(\Omega\times\mathbb{R}^d\setminus\{0\})\;|\;\exists y\in\textrm{supp}(v)\;\textrm{for which}\;(x,y,\xi,0)\in\mathrm{WF}^\alpha(v)\},$$
    where $\pi_*$ is the push-forward map.
\end{proposition}

\begin{proof}
	Since $v$ is compactly supported and since the action of $\pi_*$ is tantamount to a partial evaluation against the constant function $1\in C^\infty(\Omega^\prime)$, {\it i.e.}, $\pi_*(v)(\phi)=v(\phi\otimes 1)$ for all $\phi\in\mathcal{E}(\Omega)$, then $\pi_*(v)\in\mathcal{E}^\prime(\Omega)$. On account of Proposition \ref{prop:property besov wavefront}, a pair $(x,\xi)\in\mathrm{WF}^\alpha(\pi_*v)$ if and only if $(x,\xi)\in WF(\pi_*v - u)$ for all $u\in B^\alpha_{\infty,\infty}(\Omega)\cap\mathcal{E}^\prime(\Omega)$. Here we can restrict the attention to compactly supported elements lying in $B^\alpha_{\infty,\infty}(\Omega)$ since $\pi_*v\in\mathcal{E}^\prime(\Omega)$.
	
	In turn, on account of Corollary \ref{Cor: projection on Besov spaces}, we can replace $u$ by $\pi_*(\tilde{u})$, where $\tilde{u}\in B^\alpha_{\infty,\infty}(\Omega\times\Omega^\prime)\cap\mathcal{E}^\prime(\Omega\times\Omega^\prime)$. In other words it turns out that 
	$$(x,\xi)\in\mathrm{WF}^\alpha(\pi_*v)\Longleftrightarrow (x,\xi)\in WF(\pi_*(v-\tilde{u}))\quad\forall\tilde{u}\in B^\alpha_{\infty,\infty}(\Omega\times\Omega^\prime)\cap\mathcal{E}^\prime(\Omega\times\Omega^\prime).$$
	Applying \cite[Thm. 8.2.12]{Hor03}, it descends that $$WF(\pi_*(v-\tilde{u}))\subseteq\{(x,\xi)\;|\;\exists y\in\textrm{supp}(v-\tilde{u})\;\textrm{for which}\; (x,y,\xi,0)\in WF(v-\tilde{u})\}.$$
	Yet, on account of the arbitrariness of $\tilde{u}$ and using Proposition \ref{prop:property besov wavefront}, it descends that $y\in\mathrm{supp}(v)$ and $(x,y,\xi,0)\in\mathrm{WF}^\alpha(v)$, which is nothing but the sought statement. 
\end{proof}

\noindent We can prove the main result of this part of our work and we divide it in two statements.

\begin{theorem}\label{Thm: WF and Kernels - Dappia version}
	Let $\Omega \subseteq \mathbb{R}^n, \Omega^\prime \subseteq \mathbb{R}^m$ be open subsets, $K \in \mathcal{D}^\prime(\Omega\times\Omega^\prime)$ be the kernel of $\mathcal{K}\colon C_0^\infty(\Omega^\prime) \to \mathcal{D}^\prime(\Omega)$. Then, for all $\alpha\in\mathbb{R}$ and for all $u\in C^\infty_0(\Omega^\prime)$,
	$$\mathrm{WF}^\alpha(\mathcal{K}(u))\subset\{(x,\xi)\;|\;\exists y\in\textrm{supp}(u)\;\textrm{for which}\;(x,y,\xi,0)\in\mathrm{WF}^\alpha(K)\}.$$
\end{theorem}

\begin{proof}
	Let $\pi:\Omega\times\Omega^\prime\to\Omega$ be the projection map on the second factor and assume for the time being that $K\in\mathcal{E}^\prime(\Omega\times\Omega^\prime)$. It descends that $\mathcal{K}(u)=\pi_*(K\cdot(1\otimes u))$, where $\pi_*$ is the push-forward along $\pi$ while $\cdot$ stands for the product of distributions. Observe that, since $\mathrm{WF}^\alpha(1\otimes u)=\emptyset$ for all $\alpha\in\mathbb{R}$ then, the pointwise product is well-defined on account of Theorem \ref{th:product theorem 2}. The latter also entails that, for all $\alpha\in\mathbb{R}$,
	$$\mathrm{WF}^\alpha(K\cdot(1\otimes u))\subset\{(x,y,\xi,\eta)\in \mathrm{WF}^\alpha(K)\;|\;y\in\textrm{supp}(u)\}.$$
	At this stage, observing that by localizing the underlying distribution around each point of the wavefront set, we can apply Proposition \ref{Prop: ancillary}. It descends
	$$\mathrm{WF}^\alpha(\pi_*(K\cdot(1\otimes u)))\subset\{(x,\xi)\;|\;\exists y\in\textrm{supp}(u)\;\textrm{for which}\;(x,y,\xi,0)\in\mathrm{WF}^\alpha(K)\},$$
	which concludes the proof.
\end{proof}

\noindent At last we generalize the preceding theorem so to investigate under which circumstances $u$ can be taken to be an element lying $\mathcal{E}^\prime(\Omega^\prime)$ and with a non empty wavefront set. 

\begin{theorem}\label{Thm: WF and Kernels - Dappia version extended}
	Let $\Omega \subseteq \mathbb{R}^n, \Omega^\prime \subseteq \mathbb{R}^m$ be open subsets, $K \in \mathcal{D}^\prime(\Omega\times\Omega^\prime)$ be the kernel of $\mathcal{K}\colon C_0^\infty(\Omega^\prime) \to \mathcal{D}^\prime(\Omega)$ and $u \in\mathcal{E}^\prime(\Omega^\prime)$. In addition, for any $\alpha\in\mathbb{R}$, we call 
	\begin{equation}\label{Eq: Projected WF}
	-\mathrm{WF}^{\alpha}_{\Omega^\prime}(K):=\{(y,\eta)\in\Omega^\prime\times\mathbb{R}^m\setminus\{0\}:\;\exists x\in\Omega\;|\; (x,y,0,-\eta) \in \mathrm{WF}^\alpha(K)\}.
	\end{equation}
	If for any $(y,\eta) \in \Omega^\prime\times(\mathbb{R}^m\setminus\{0\})$ there exists $\alpha_1,\alpha_2 \in \mathbb{R}$ with $\alpha_1 + \alpha_2 > 0$ such that
	\begin{equation}\label{Eq: kernel condition}
		(y,\eta) \not \in -\mathrm{WF}^{\alpha_1}_{\Omega^\prime}(K) \cup \mathrm{WF}^{\alpha_2}(u),
	\end{equation}
	then there exists $\mathcal{K}u \in \mathcal{D}^\prime(\Omega)$. Furthermore, if $\alpha\leq\alpha_1+\alpha_2$, then
	$$WF^\alpha(\mathcal{K}(u))\subseteq\{(x,\xi)\in\Omega\times(\mathbb{R}^n\setminus\{0\}):\exists (y,\eta) \in \Omega^\prime \times (\mathbb{R}^m\setminus\{0\}) | (x,y,\xi,\eta) \in X\cup Y\},$$
	where 
	$$X:=\{(x,y,\xi,\eta)\in\mathrm{WF}^{\alpha_1}(K)\;|\;(y,-\eta)\in WF_0(u)\},\quad Y=\{(x,y,\xi,\eta)\in WF(K)\;|\;(y,-\eta)\in\mathrm{WF}^{\alpha_2}(u)\}.$$
\end{theorem}

\begin{proof}
Following the same strategy as in the proof of Theorem \ref{Thm: WF and Kernels - Dappia version}, we aim at writing $\mathcal{K}(u):=\pi_*(K\cdot(1\otimes u))$ where $\pi_*$ is the push-forward built out of the projection map $\pi:\Omega\times\Omega^\prime\to\Omega$. Given $\alpha_2\in\mathbb{R}$, Equation \eqref{Eq: WF-Tensor Product_improved} entails that $\mathrm{WF}^{\alpha_2}(1\otimes u)\subseteq(\textrm{supp}(u)\times 0)\times\mathrm{WF}^{\alpha_2}(u)$, which combined with Theorem \ref{th:product theorem 2} and Equation \eqref{Eq: kernel condition}, entails that there exists $K\cdot(1\otimes u)\in\mathcal{D}^\prime(\Omega\times\Omega^\prime)$. Yet, being $u$ compactly supported we can act with the push-forward along the map $\pi:\Omega\times\Omega^\prime\to\Omega$, hence obtaining that $\pi_*(K\cdot (1\otimes u))\in\mathcal{D}^\prime(\Omega)$. 

A straightforward adaptation to the case in hand of Proposition \ref{Prop: ancillary} entails that, for every $\alpha\in\mathbb{R}$, $\mathrm{WF}^\alpha(\pi_*(K\cdot(1\otimes u)))$ is contained within the collection of points $(x,\xi)\in\Omega\times\mathbb{R}^d\setminus\{0\}$ for which there exists $y\in\Omega^\prime$ such that $(x,y,\xi,0)\in\mathrm{WF}^\alpha(K\cdot(1\otimes u))$.

Suppose now that $\alpha=\alpha_1+\alpha_2$. Theorem \ref{th:product theorem 2}, Equation \eqref{Eq: WF estimate product} in particular entails that the collection of points $(x,y,\xi,0)\in\mathrm{WF}^\alpha(K\cdot(1\otimes u))$ is contained in those of the form $(x,y,\xi,0)$ such that one of the two following conditions is met:
\begin{enumerate}
	\item there exists $\eta \in \mathbb{R}^m$ such that $(x,y,\xi,\eta)\in\mathrm{WF}^{\alpha_1}(K)$ and $(y,-\eta)\in WF_0(u)$, 
	\item there exists $\eta \in \mathbb{R}^m$ such that $(x,y,\xi,\eta)\in WF(K)$ and $(y,-\eta)\in\mathrm{WF}^{\alpha_2}(u)$.
\end{enumerate}
To conclude it suffices to recall that, on account of Equation \eqref{eq:inclusion besov spaces} $\mathrm{WF}^\alpha(K\cdot(1\otimes u))\subseteq\mathrm{WF}^{\alpha_1+\alpha_2}(K\cdot(1\otimes u))$ whenever $\alpha\leq\alpha_1+\alpha_2$.
\end{proof}

\noindent To conclude, we prove a statement which adapts to the current scenario an important result for the Sobolev wavefront set, see \cite[Prop. B.9]{Junker:2001gx}. 

\begin{corollary}\label{Cor: WF and Kernels - Fede version}
	Let $\Omega \subseteq \mathbb{R}^n, \Omega^\prime \subseteq \mathbb{R}^m$ be open subsets, $K \in \mathcal{D}^\prime(\Omega\times\Omega^\prime)$ be the kernel of $\mathcal{K}\colon C_0^\infty(\Omega^\prime) \to \mathcal{D}^\prime(\Omega)$ and $u \in\mathcal{E}^\prime(\Omega^\prime)$. 
	Assume in addition that for any $(y,\eta) \in \Omega^\prime\times(\mathbb{R}^m\setminus\{0\})$ there exists $\alpha_1,\alpha_2 \in \mathbb{R}$ with $\alpha_1 + \alpha_2 > 0$ such that
	\begin{equation}\label{eq:condition kernel}
		(y,\eta) \not \in -WF^{\alpha_1}_{\Omega^\prime}(K) \cup WF^{\alpha_2}(u).
	\end{equation}
	If $WF_{\Omega^\prime}(K) =\emptyset$ and if there exists $\gamma\in\mathbb{R}$ such that $\mathcal{K}(B_{\infty,\infty}^\alpha(\Omega^\prime)\cap\mathcal{E}^\prime(\Omega^\prime)) \subset B_{\infty,\infty}^{\alpha-\gamma,\mathrm{loc}}(X)$, then
	\begin{equation}\label{eq:inclusion convolution}
		\mathrm{WF}^{\alpha-\gamma}(\mathcal{K}u) \subseteq WF^\prime(K) \circ \mathrm{WF}^\alpha(u) \cup WF_\Omega(K),
	\end{equation}
	where $WF^\prime(K) \circ \mathrm{WF}^\alpha(u) := \{(x,\xi)\;|\;\exists (y,\eta) \in \mathrm{WF}^\alpha(u)\;\textrm{for which}\; (x,y,\xi,-\eta)\in WF(K)  \}$ while $W_\Omega(K):=\{(x,\xi)\in\Omega\times\mathbb{R}^n\;:\;\exists y\in\Omega^\prime\;|\;(x,y,\xi,0)\in WF(K)\}$.
\end{corollary}

\begin{proof}
	On account of Theorem \ref{Thm: WF and Kernels - Dappia version extended}, Equation \eqref{eq:condition kernel} entails that $\mathcal{K}(u)\in\mathcal{D}^\prime(\Omega)$. Bearing in mind Proposition \ref{prop:property besov wavefront}, we can find an open conic neighborhood $\Gamma\subset \mathrm{WF}^\alpha(u)$ such that $WF(u-v) \subset \Gamma$ for all $v \in B_{\infty,\infty}^{\alpha,\mathrm{loc}}(\Omega^\prime)$. Per assumption $\mathcal{K}(v)\in B_{\infty,\infty}^{\alpha-\gamma,\mathrm{loc}}(\Omega)$, which entails in turn on account of \cite[Theorem 8.2.13]{Hor03}
	\[
	\mathrm{WF}^{\alpha-\gamma}(\mathcal{K}u) \subseteq WF(\mathcal{K}(u-v)) \subseteq WF^\prime(K) \circ WF(u-v) \cup WF_X(K) \subset WF^\prime(K) \circ \Gamma \cup WF_X(K).
	\]
	To conclude, in view of the arbitrariness of $\Gamma$, we infer
	\[
	\mathrm{WF}^{\alpha-\gamma}(\mathcal{K}u) \subseteq WF^\prime(K) \circ \mathrm{WF}^\alpha(u) \cup WF_\Omega(K).
	\]
\end{proof}

\begin{example}
	Let us consider the heat kernel operator, namely the fundamental solution of the heat equation $G\in\mathcal{D}^\prime(\mathbb{R}^{d+1}\times\mathbb{R}^{d+1})$, whose integral kernel reads in standard Cartesian coordinates
	\[
	G(t,x,t^\prime,x^\prime)= \frac{\Theta(t-t^\prime)}{(4\pi (t-t^\prime))^{d/2}} e^{-\frac{\lvert x-x^\prime \rvert^2}{4(t-t^\prime)}},
	\]
	where $\Theta$ is the Heaviside function.
	By Schauder estimates, \emph{c.f.} \cite{Sim97},  $G$ can also be read as the kernel of an operator $\mathcal{G} \colon B^\alpha_{\infty,\infty}(\mathbb{R}^{1+d}) \to B^{\alpha + 2}_{\infty,\infty}(\mathbb{R}^{1+d})$. Furthermore it holds that
	\begin{equation}\label{Eq: WF of heat kernel}
		WF(G)=\{(t,x,t,x,\tau,\xi,-\tau,-\xi)\;|\; (t,x)\in\mathbb{R}^{d+1}\;\textrm{and}\; (\tau,\xi)\in\mathbb{R}^{d+1}\setminus\{0\}\}. 
	\end{equation}
	Therefore, we are in position to apply\eqref{eq:inclusion convolution}. Considering any $u\in \mathcal{E}^\prime(\mathbb{R})$, we can infer that the hypotheses of Corollary \ref{Cor: WF and Kernels - Fede version} are met since $\mathrm{WF}^\alpha_{\mathbb{R}^{d+1}}(G)=\emptyset$ for all $\alpha\in\mathbb{R}$, where the subscript $\mathbb{R}^{d+1}$ should be read in the sense of Equation \eqref{Eq: Projected WF}. At the same time, on account of Remark \ref{Ex: Besov WF of compactly supported distributions}, there must exist $\alpha<0$ such that $\mathrm{WF}^\alpha(u)=\emptyset$. This entails that
	\[
	\mathrm{WF}^{\alpha+2}(\mathcal{G}(u)) \subseteq WF^\prime(G) \circ \mathrm{WF}^\alpha(u),
	\]
	which, combined with Equation \eqref{Eq: WF of heat kernel}, yields $WF^\prime(G) \circ \mathrm{WF}^\alpha(u) = \mathrm{WF}^\alpha(u)$. This leads to the inclusion
	\[
	\mathrm{WF}^{\alpha+2}(\mathcal{G}u) \subseteq  \mathrm{WF}^\alpha(u).
	\]
\end{example}

\subsection{Besov Wavefront Set and Hyperbolic Partial Differential Equations}

As an application of the results of the previous sections, we study the interplay between the Besov wavefront set and a large class of hyperbolic partial differential equations of the form
\begin{equation}\label{eq: hyperbolic system}
	\partial_t u = i a(D_x) u, \quad (t,x) \in \mathbb{R}\times \mathbb{R}^d,
\end{equation}
where we assume $a=a_1+a_0$ where $a_1\in S^1_{\textrm{hom}}(\mathbb{R}^d)$, while $a_0\in S^0(\mathbb{R}^d)$ see Definition \ref{Def: Symbols}. Using standard Fourier analysis, we can infer that the fundamental solution associated to the operator $\partial_t-ia(D_x)$ is the distribution $G\in\mathcal{D}^\prime(\mathbb{R}\times\mathbb{R}^d)$, whose integral kernel reads
\[
G(t,x) = \Theta(t) [e^{it a(D)} \delta](x),
\]
where $\Theta$ is once more the Heaviside function. 

\begin{proposition}\label{prop:inclusion besov spaces}
		Let $\alpha \in \mathbb{R}$. Then $B_{\infty,\infty}^\alpha(\mathbb{R}^d) \cap \mathcal{E}^\prime(\mathbb{R}^d)\subset B_{2,\infty}^\alpha(\mathbb{R}^d)$.
\end{proposition}

	\begin{proof}
		Let $v \in B_{\infty,\infty}^\alpha(\mathbb{R}^d) \cap \mathcal{E}^\prime(\mathbb{R}^d)$. For any $\kappa \in \mathscr{B}_{\lfloor \alpha \rfloor}$ as per Definition \ref{Def: Kernel for Besov}, it 
		\begin{align}
			\|v(\kappa^\lambda_{x})\|_{L^2(\mathbb{R}^d)} \lesssim \|v(\kappa^\lambda_{x})\|_{L^\infty(\mathbb{R}^d)} \lesssim \lambda^\alpha, \nonumber
		\end{align}
		where the first estimate is a a byproduct of $v$ being compactly supported. A similar reasoning applies when considering any $\underline{\kappa}\in \mathcal{D}(B(0,1))$ such that $\check{\kappa}(0) \neq 0$. As a consequence of Definition \ref{def:definition Besov local means}, we infer that $v \in B_{2,\infty}^\alpha(\mathbb{R}^d)$.		
\end{proof}

\begin{proposition}\label{Prop: regularity of G}
	Let $G \in \mathcal{D}^\prime(\mathbb{R}\times \mathbb{R}^d)$ be the fundamental solution of the hyperbolic operator $\partial_t - i a(D_x)$. Then, $G(t,\cdot) \in B_{2,\infty}^{-\frac{d}{2}}(\mathbb{R}^d)$ for any $t \in \mathbb{R}$. Moreover, given $v \in B_{\infty,\infty}^{\alpha,\mathrm{loc}}(\mathbb{R}^d)$ with $\alpha \in \mathbb{R}$, 
	\[
	G(t,\cdot) \ast v \in B_{\infty,\infty}^{\alpha-\frac{d}{2},\mathrm{loc}}(\mathbb{R}^d),
	\]
where $\ast$ stands for the convolution.
\end{proposition}
\begin{proof}
	Let $\{\psi_j\}_{j\geq 0}$ be a Littlewood-Paley partition of unity as per Definition \ref{Def: L-P partition of unity}. For any $j\geq 1$, it descends
	\begin{align}
		\|\psi_j(D_x)e^{ita(D_x)}\delta\|_{L^2(\mathbb{R}^d)} = \|\psi_j\|_{L^2(\mathbb{R}^d)} = 2^{j \frac{d}{2}} \|\psi\|_{L^2(\mathbb{R}^d)},
	\end{align}
	where we applied Fourier-Plancherel theorem in the first equality. Hence we can conclude that
	\[
	\sup_{j\geq 0}2^{-j \frac{d}{2}} \|\psi_j(D_x)e^{ita(D_x)}\delta\|_{L^2(\mathbb{R}^d)} < \infty,
	\]
	which entails that $G(t,\cdot) \in B_{2,\infty}^{-\frac{d}{2}}(\mathbb{R}^d)$.
	Observe that, for every $\phi \in \mathcal{D}(\mathbb{R}^d)$, $\phi v \in B^\alpha_{2,\infty}(\mathbb{R}^d)$ on account of Proposition \ref{prop:inclusion besov spaces}. Then, as a consequence of \cite[Thm 2.2]{Kuhn}, we can infer that $G(t,\cdot) \ast (\phi v) \in B_{\infty,\infty}^{\alpha-\frac{d}{2}}(\mathbb{R}^d)$ for any $t \in\mathbb{R}$.
\end{proof}

Proposition \ref{Prop: regularity of G} can be read as a statement that the solution map associated to Equation \eqref{eq: hyperbolic system}
\[
S(t,0) : u(0) \mapsto u(t)
\]
is continuous from $B_{\infty,\infty}^{\alpha,\mathrm{loc}}(\mathbb{R}^d)$ to $B_{\infty,\infty}^{\alpha-\frac{d}{2},\mathrm{loc}}(\mathbb{R}^d)$. Moreover, $S(t,0)$ can be inverted and $S(t,0)^{-1}=S(0,t)$.

\begin{theorem}\label{th:propagation of singularities}
	Let $a$ be as per Equation \eqref{eq: hyperbolic system} and let $u_0 \in \mathcal{S}^\prime(\mathbb{R}^d)$. Suppose that $u$ is the solution of the initial value problem
	\begin{equation}
		\begin{cases}
			\partial_t u = i a(D_x) u,\\
			u(0)= u_0.
		\end{cases}
	\end{equation}
	Then, for every $\alpha\in\mathbb{R}$,
	\begin{equation}
		\mathrm{WF}^{\alpha-\frac{d}{2}}(u(t)) = \mathcal{C}(t) \mathrm{WF}^\alpha(u_0),
	\end{equation}
	where $\mathcal{C}(t)$ is the flow from $t$ to $0$ associated the Hamiltonian vector field $H_{a(\xi)}$.
\end{theorem}
\begin{proof}
	We just prove the inclusion $\subset$, the other following suite. Let us consider $(x,\xi)\not \in \mathrm{WF}^\alpha(u_0)$. Then there exists $A \in \Psi^0(\mathbb{R}^d)$, elliptic at $(x,\xi)$, such that $A u_0 \in B_{\infty,\infty}^{\alpha,\mathrm{loc}}(\mathbb{R}^d)$. Let us define $A(t):=S(t,0)\circ A\circ S(0,t)$ so that $A(t)u(t) = S(t,0)Au_0 \in B_{\infty,\infty}^{\alpha-\frac{d}{2},\mathrm{loc}}(\mathbb{R}^d)$. On account of Egorov's theorem, see {\it e.g.} \cite{Hintz}, we can conclude that $A(t)$ still lies in $\Psi^0(\mathbb{R}^d)$ and it is elliptic at $\mathcal{C}(t)^{-1}(x,\xi)$. This implies $\mathcal{C}(t)^{-1}(x,\xi) \not \in \mathrm{WF}^{\alpha-\frac{d}{2}}(u(t))$.
\end{proof}

\begin{remark}
	It is worth mentioning that the estimate on the Besov wavefront set as per Theorem \ref{th:propagation of singularities} might be improved if working with a generic Besov space $B^\alpha_{pq}(\mathbb{R}^d)$ rather than with $B^\alpha_{\infty\infty}(\mathbb{R}^d)$. Yet this step requires first of all to establish an improved version of Proposition \ref{Prop: regularity of G}, which appears to be elusive at this stage.
\end{remark}

\paragraph{Acknowledgements} 

We are thankful to M. Capoferri and N. Drago for helpful discussions and comments. The work of F.S. is supported by a scholarship of the University of Pavia, while that of P.R. by a fellowship of the Instiute for Applied Mathematics of the University of Bonn.

\end{document}